\documentclass[twocolumn]{autart} 
\usepackage{tikz}
\usepackage{tikz-network}
\usepackage{amsmath}
\usepackage{amssymb}
\usepackage{amsfonts}
\usepackage{graphicx}
\usepackage{enumitem}
\usepackage{mathrsfs}
\usepackage{float}
\usepackage{cite}
\usepackage{soul}
\usepackage{tikz}
\usetikzlibrary{matrix,arrows,decorations.pathmorphing}
\usetikzlibrary{decorations.markings}
\usepackage{verbatim}
\usepackage{mathtools}
\usepackage{caption}
\usepackage{subcaption}
\usetikzlibrary{arrows}
\usepackage{tabularx}
\usepackage{xcolor}
\usepackage{inputenc}
\newtheorem{proposition}{Proposition}
\newtheorem{corollary}{Corollary}
\newtheorem{example}{Example}
\newtheorem{theorem}{Theorem}
\newtheorem{lemma}{Lemma}
\newtheorem{remark}{Remark}
\newtheorem{definition}{Definition}
\usepackage{array}
\usepackage{textcomp}
\usepackage{stfloats}
\usepackage{url}
\usepackage{verbatim}
\usepackage{graphicx}
\usepackage{cite}
\usepackage{blkarray}

\providecommand{\qedsymbol}{$\square$}
\newenvironment{proof}{\par\noindent\textbf{Proof. }\ignorespaces}{\hfill\qedsymbol\par}

\begin{document}
\begin{frontmatter}

\title{Invariance of Competition Outcomes in Hypergraph Competitive Dynamics}

\author[qksj]{Qi Zhao}\ead{zhaoqi\_1995@163.com}, 
\author[rugjbi]{Shaoxuan Cui}\ead{s.cui@rug.nl},    
\author[qkzd]{Baolin Zhang}\ead{zhangbl2006@163.com},
\author[qksj]{Junwei Du}\ead{djwqd@163.com},
\author[xidian]{Yuanshi Zheng}\ead{zhengyuanshi2005@163.com},  

\thanks[footnoteinfo]{The corresponding author is S. Cui. This work is partially financially supported by the Qingdao Natural Science Foundation (Grant No. 25-1-1-121-zyyd-jch), the opening project of Shandong Province Engineering Research Centre (Qingdao University of Science and Technology, Grant No.KF2024SD010), and the NSFC (Grant No. 62401310).}

\address[qksj]{School of Data Science, Qingdao University of Science and Technology, China} 
\address[rugjbi]{Bernoulli Institute, Faculty of Science and Engineering, University of Groningen, The Netherlands}  
\address[qkzd]{College of Automation and Electronic Engineering, Qingdao University of Science and Technology, China}             
\address[xidian]{Shaanxi Key Laboratory of Space Solar Power Station System, School of Mechano-electronic Engineering, Xidian University, China}        


\begin{abstract}                          
	Winner-take-all (WTA)--type selection is a fundamental mechanism in networked competition, yet its dependence on higher-order interactions remains insufficiently understood. We study a Lotka--Volterra competitive dynamics on higher-order networks, where classical pairwise inhibition is augmented by multi-way interaction terms induced by hyperedges of uniform hypergraphs. The proposed model shows multiple competitive outcomes, including WTA, winner-share-all (WSA), and variant winner-take-all (VWTA). The existence, uniqueness and stability of equilibria are rigorously proved through mathematical analysis, which relies on classical stability theory and recent advances in tensor algebra. We show that the eventual selection outcome is relatively insensitive to the hyperedge order and the specific higher-order coupling structure, and is instead determined by a small set of interpretable scalar parameters, such as the ratio between self-inhibition and lateral-inhibition and the external inputs. Numerical experiments support the theory by showing that higher-order interactions affect convergence and steady states, yet yield the similar outcome taxonomy (WTA/WSA/VWTA) as in standard graphs. These results provide a network-scientific explanation of the robustness of WTA-type outcomes under complex group interactions and offer principled guidance for designing selection mechanisms on higher-order networks.
\end{abstract}

\end{frontmatter}

\section{Introduction}
WTA--type selection is a fundamental mechanism in networked competition, widely underlying resource allocation and selective activation processes in complex systems. In many real settings, interactions are not purely pairwise. Group-level non-additive couplings arise from collective contexts such as group communication, co-activity, and shared constraints. This motivates competitive dynamics on higher-order networks such as hypergraphs. Neural systems provide a representative and well-studied instance of such networked competition: inhibitory interactions among neurons are known to shape selective activation \cite{Han-2007}, decision-making \cite{Tang-2024}, and cognitive behavior \cite{Clark-2004}. Unlike generic dynamical systems, competitive neural dynamics often exhibits strong selection effects such as WTA and WSA, making it a natural testbed for developing rigorous theory on competitive outcomes in higher-order networked systems.

To model these competitive processes, researchers have drawn inspiration from ecological systems, in which species compete and interact in dynamic environments. The Lotka-Volterra (LV) model, initially proposed to depict predator--prey interactions \cite{Lotka-1920} and interspecies competition \cite{Volterra-1928}, has been introduced into neural networks to describe neuronal inhibition \cite{Xiao-2017} and selection \cite{Pavel-2011}. The paper \cite{Asai-1999} constructed LV-based neural circuits that exhibit WTA and WSA behaviors. Zhang and Yi \cite{ZhangYi-2002} investigated the stability and multi-stability of LV neural networks with time delays. The paper \cite{Es-saiydy-2022} further extended the model by incorporating time-varying and distributed delays, and explored the existence and stability of system solutions. In addition, the above model has been used to simulate metastable transitions in neural systems, providing theoretical insights into cognitive processes \cite{Bick-2017}. Fukai and Tanaka \cite{Fukai-97} demonstrated that LV-type equations derived from membrane dynamics can effectively model neuronal competition, revealing distinct behaviors such as WTA and WSA. Building on this, \cite{Yi-2013} applied LV neural networks to fMRI-based brain region detection. Zheng proposed a WTA LV recurrent neural network with $N\times N$ neurons and established sufficient conditions for stable equilibria with a unique winner per row/column in \cite{Bochuan-Zheng}. More recently, data-driven approaches using neural differential equations \cite{Devgupta-2024} and spatial models with non-local interactions \cite{Li-2021} have further expanded the applicability of LV model. The LV model combines simplicity and computational efficiency, which makes it suitable for simulating the dynamic evolution of neural populations. It has become a key tool for studying neuronal competition and is widely applied in computational neuroscience.

Despite these advances, most existing LV-based neural models are restricted to pairwise interaction\cite{Lotka-1920}-\cite{Li-2021}, typically represented by simple graphs. However, increasing evidence shows that real neural circuits involve complex dependencies among groups of neurons \cite{Abrams-1983}. To more accurately capture these higher-order interactions, researchers have adopted hypergraph structures. In a hypergraph, a single edge can connect more than two nodes, better reflecting the collective behavior of neural groups. 
Higher-order structures have been shown to drastically alter system dynamics in various fields. In \cite{Cui-SIS}, the authors demonstrated that hypergraph structures fundamentally alter the diffusion dynamics of SIS processes, inducing multi-stability phenomena absent in pairwise interaction models. For instance, synchronization models on hypergraphs can exhibit abrupt transitions, multistability, and explosive behaviors not seen in pairwise interaction systems. For example, \cite{Mayfield2017} proposed a simple framework to incorporate non-additive higher-order interactions into biodiversity models and confirms their significant impact on species performance in natural plant communities. Carletti and Fanelli proposed a general framework for modeling dynamics on hypergraphs in \cite{Carletti2020}, and showed that higher-order interactions can alter stability, conservation properties, and the structure of equilibria. Battiston et al. \cite{Battiston2021} showed that such structures can change system dynamics like synchronization and diffusion. From a computational perspective, Cui et al. \cite{Cui-DT2024, Cui2025Metzler} employed tensor-based Perron--Frobenius methods to analyze both discrete- and continuous-time systems on hypergraphs, providing rigorous frameworks for attractor characterization and stability analysis of high-order interaction dynamics. Notably, although higher-order interactions commonly impact system behavior, this paper presents a competition model wherein such interactions enhance modeling fidelity without altering the ultimate outcome or influencing factors of the WTA mechanism, which indicates that the outcome pattern maybe one of the natural characteristics of a homogeneous network competition.

Building on the modeling advances in higher-order networks and hypergraphs, researchers have extended the LV model to incorporate higher-order interactions. These extensions aim to capture non-additive effects and explain complex dynamics observed in empirical systems. Higher-order interactions that capture collective effects among three or more species have been shown to enhance predictive accuracy in a wide range of ecological conditions. For example, \cite{Letten-2019} extended the LV competition model with higher-order interactions and demonstrated through data and simulations their prevalence and contribution to predictive accuracy. In \cite{Sidhom-2024}, Sidhom and Galla introduced a stochastic version with random interaction tensors. Their results showed better system stability and diversity. At the same time, Sales-Pardo et al. \cite{Sales-Pardo-2023} proposed a generative model for predicting hyperedges. They found statistical links between higher-order and lower-order interactions. However, the inherent nonlinearity of these models makes their analytical treatment challenging. The influence of higher-order interactions on species coexistence has been examined through simulations \cite{Singh-2021} and through numerical and statistical physics approaches \cite{Gibbs-2022}. At present, rigorous mathematical proofs regarding the properties of equilibria remain limited. Recently, Cui et al. \cite{cui2024lotka} applied S-tensor theory to high-order LV systems, establishing equilibrium conditions via a polynomial complementarity formulation. While providing a tensor-based framework for nonlinear dynamics, their analysis does not offer detailed and comprehensive results for competitive scenarios.

Most existing models still fail to adequately explore higher-order dependencies in neural competition. In particular, much of the literature focuses on pairwise couplings, and the role of higher-order interactions in shaping competition dynamics remains underexplored. Moreover, the nonlinear and high-dimensional structure of such systems complicates rigorous analysis of equilibrium existence, uniqueness, and stability. To address these gaps, this paper proposes a higher-order LV competition model on hypergraphs. The main contributions of this work are summarized as follows:
\begin{enumerate}
	\item[(i)] We propose a high-order LV competition model based on uniform hypergraphs, which extends traditional pairwise interaction models to capture multi-neuron interactions in higher-order and even arbitrary-order network structures.
	\item[(ii)] We rigorously analyze the existence, uniqueness, and stability of equilibria under several competitive outcomes, including all-neurons coexistence, WSA, WTA, and VWTA.
	\item[(iii)] We obtain explicit and verifiable conditions for these results by leveraging properties of structured tensors.
	\item[(iv)] Numerical simulations corroborate the theory and demonstrate that higher-order couplings can affect convergence behavior and steady-state values, while preserving the invariant outcome under the same inputs and inhibition ratio, indicating that higher-order interactions enable a more realistic representation of neural coordination beyond pairwise couplings.
	\item[(v)] Importantly, we identify a structurally invariant competition outcome pattern in our model: while higher-order interactions modulate the dynamics, the fundamental competition outcome (WTA/WSA/VWTA) and its determining factors remain relatively robust. This indicates that the ratio between self-inhibition and lateral inhibition plays a more decisive role than the hyperedge order and the detailed interaction structure.
\end{enumerate}

The rest of this paper is structured as follows. The preliminaries about hypergraph, structured tensor theory and effective lemmas are introduced in Section \ref{preliminaries}. Section \ref{existence}-\ref{Application} contains the problem models and main theoretical results. In Section \ref{Simulations}, some numerical simulations are proposed to verify the validity of the designed algorithms. In the end, section \ref{conclusion} summarizes the main contents of the paper.

\section{PRELIMINARIES}\label{preliminaries}

This section introduces the foundational concepts and mathematical tools necessary for modeling and analyzing High-order competition mechanisms based on the LV model. We begin by presenting the hypergraph framework used to characterize higher-order interactions, followed by the tensor representations that support such modeling. Finally, we summarize a collection of relevant theoretical results related to stability analysis and tensor equations.

\subsection{Hypergraphs for High-Order Interaction Modeling}\label{Hypergraphs}
Hypergraphs generalize conventional graphs by allowing edges to connect more than two nodes simultaneously, making them a powerful tool for modeling higher-order interactions. We adopt the hypergraph formulation introduced in \cite{gallo1993directed}.

In this paper, we focus on uniform hypergraphs, where each hyperedge connects exactly $r$ distinct nodes. An $r$-uniform hypergraph is defined as a pair $\mathscr{H} = (\mathcal{V}, \mathcal{E})$, where $\mathcal{V}$ is the set of nodes and $\mathcal{E} \subseteq \binom{\mathcal{V}}{r}$ is the set of hyperedges, each containing exactly $r$ elements.

To incorporate interaction weights, we associate a symmetric order-$r$ tensor $A = [A_{i_1 i_2 \dots i_r}]$ with the hypergraph, where each entry represents the strength of the interaction among the corresponding $r$ nodes. When $r = 2$, this formulation reduces to a standard graph with an adjacency matrix $A_{ij}$.




\subsection{Tensors and Their Spectral Properties}\label{Tensors}

Tensors provide a natural algebraic framework for representing uniform hypergraph interactions. A tensor $A \in \mathbb{R}^{n \times n \times \cdots \times n} = \mathbb{R}^{[m, n]}$ is a multidimensional array of order $m$ and dimension $n$. In this paper, we consider only cubical and supersymmetric tensors, where all modes have the same dimension and the entries are invariant under any permutation of indices.

The identity tensor $\mathcal{I} = (\delta_{i_1 \cdots i_m})$ is defined elementwise as
\begin{equation*}
	\delta_{i_1 \cdots i_m} =
	\begin{cases}
		1, & \text{if } i_1 = i_2 = \cdots = i_m, \\
		0, & \text{otherwise}.
	\end{cases}
\end{equation*}

For vectors $A x^{m-1}$ and $x^{[m-1]}$, whose $i$-th components are
\begin{equation*}
	\begin{aligned}
		(A x^{m-1})_i & = \sum_{i_2, \ldots, i_m=1}^n A_{i, i_2, \ldots, i_m} x_{i_2} \cdots x_{i_m}, \\
		(x^{[m-1]})_i & = x_i^{m-1}.
	\end{aligned}
\end{equation*}
where $m$ denote the order of the tensor $A$.

We now turn our attention to the eigenvalue-eigenvector problem for tensors:
\begin{equation}\label{eq:eigenproblem}
	A x^{m-1} = \lambda x^{[m-1]},
\end{equation}
If there exists a real number \( \lambda \) and a nonzero real vector \( x \) that satisfy \eqref{eq:eigenproblem}, then \( \lambda \) is termed an H-eigenvalue of \( A \), and \( x \) is the corresponding H-eigenvector of \( A \), where \( \lambda \in \mathbb{R} \) and \( x \in \mathbb{R}^n \setminus \{0\} \). Throughout this paper, the terms eigenvalue and eigenvector are used synonymously with H-eigenvalue and H-eigenvector, respectively. The spectral radius of $A$ is
\[
\rho(A) = \max \{|\lambda| : \lambda \text{ is an eigenvalue of } A\}.
\]

\subsection{Structured Tensors and Their Properties}

In this subsection, we present several special tensor classes that are fundamental to the analysis of nonlinear tensor equations, particularly those arising from high-order interaction models. We also summarize key results ensuring the existence and uniqueness of positive solutions.

Let $A \in \mathbb{R}^{[m,n]}$ be an $m$-order and $n$-dimension real tensor. The identity tensor $\mathcal{I} \in \mathbb{R}^{[m,n]}$ is defined by
\[
\mathcal{I}_{i_1 i_2 \dots i_m} =
\begin{cases}
	1, & \text{if } i_1 = i_2 = \dots = i_m, \\
	0, & \text{otherwise}.
\end{cases}
\]

\begin{definition}[\cite{Ding2013}]
	A tensor $A \in \mathbb{R}^{[m,n]}$ is referred to as an \emph{$\mathcal{M}$-tensor} if it admits a decomposition of the form $A = s\mathcal{I} - B$, where $B \in \mathbb{R}^{[m,n]}$ is a nonnegative tensor and $s \in \mathbb{R}$ satisfies $s \geq \rho(B)$. If the inequality is strict, i.e., $s > \rho(B)$, then $A$ is called a \emph{nonsingular $\mathcal{M}$-tensor}.
\end{definition}

\begin{definition}[\cite{Cui2024}]
	A tensor $A \in \mathbb{R}^{[m,n]}$ is called a \emph{Metzler tensor} if it can be expressed as $A = s\mathcal{I} + B$, where $B$ is a nonnegative tensor and $s \in \mathbb{R}$. This implies that all off-diagonal entries of $A$ are nonnegative.
\end{definition}

\begin{definition}[\cite{wangxuezhou2019}]
	Given a tensor $A = (A_{i_1 i_2 \dots i_m}) \in \mathbb{R}^{[m,n]}$, its \emph{comparison tensor} $\langle A \rangle = (m_{i_1 i_2 \dots i_m})$ is defined componentwise by
	\[
	m_{i_1 i_2 \dots i_m} =
	\begin{cases}
		|A_{i_1 i_2 \dots i_m}|, & \text{if } (i_2, \dots, i_m) = (i_1, \dots, i_1), \\
		-|A_{i_1 i_2 \dots i_m}|, & \text{otherwise}.
	\end{cases}
	\]
	Tensor $A$ is called an \emph{$\mathcal{H}$-tensor} if its comparison tensor $\langle A \rangle$ is an $\mathcal{M}$-tensor. It is called a \emph{nonsingular $\mathcal{H}$-tensor} if $\langle A \rangle$ is a nonsingular $\mathcal{M}$-tensor. In addition, if all diagonal entries $A_{i i \dots i}$ are positive, then $A$ is called an \emph{$\mathcal{H}^{+}$-tensor}.
\end{definition}

\begin{definition}[\cite{Qi2005}]
	A tensor $A \in \mathbb{R}^{[m,n]}$ is said to be \emph{diagonally dominant} if for all $i = 1,2,\dots,n$, it holds that
	\[
	|A_{i i \dots i}| \geq \sum_{(i_2, \dots, i_m) \neq (i, \dots, i)} |A_{i i_2 \dots i_m}|,
	\]
	if the above inequality is strict, then $A$ is referred to as a strictly diagonally dominant tensor.
\end{definition}

\begin{remark}\label{Diagonal-dominance}
	Diagonal dominance implies that the eigenvalues of a tensor are mainly determined by its diagonal entries, as shown in \cite{Qi2005} through a Gershgorin-type theorem. This structural property supports the positive definiteness of the polynomial form, which is crucial for ensuring the stability of nonlinear systems. In applications such as neural dynamics or biological competition modeling, diagonal dominance reflects strong self-regulation that suppresses the influence of external interactions. 
\end{remark}

It is known that both nonsingular $\mathcal{M}$-tensors and strictly diagonally dominant tensors with positive diagonal elements are special cases of $\mathcal{H}^{+}$-tensors. These tensors enjoy favorable properties concerning the solvability of certain tensor equations. Next, we present several lemmas that establish the existence and uniqueness of positive solutions to nonlinear tensor equations under structural assumptions on the involved tensors.

\begin{lemma}[\cite{wangxuezhou2019}]\label{unique-positive-solution-1}
	Let $\mathcal{A} \in \mathbb{R}^{[m,n]}$ be a strictly diagonally dominant tensor with all diagonal entries $\mathcal{A}_{i i \dots i} > 0$. Then, for any positive vector $b \in \mathbb{R}^n$, the equation $\mathcal{A} x^{m-1} = b$ admits a unique positive solution $x \in \mathbb{R}^n$.
\end{lemma}

\begin{lemma}[\cite{wangxuezhou2019}]\label{unique-positive-solution-2}
	If $\mathcal{A} \in \mathbb{R}^{[m,n]}$ is an $\mathcal{H}^{+}$-tensor, then the equation $\mathcal{A} x^{m-1} = b$ has a unique positive solution for every $b > 0$.
\end{lemma}

It is worth mentioning that a strictly diagonally dominant tensor with all positive diagonal entries is an $\mathcal{H}^{+}$-tensor \cite{wangxuezhou2019}.

Furthermore, we develop the following results for the later system analysis.
\begin{proposition}\label{prop:all_ones_tensor_unique_pos}
	Let $\mathcal A\in\mathbb R^{[k,n]}$ be the order-$k$ ($k\ge 2$) and dimension-$n$ tensor whose entries satisfy
	$a_{i_1\cdots i_k}=1$ for all $i_1,\dots,i_k\in\{1,\dots,n\}$.
	Consider the tensor equation
	\begin{equation}\label{eq:all_ones_tensor_eq}
		\mathcal A x^{k-1}=b,
	\end{equation}
	where $b\in\mathbb R^n$ and $x\in\mathbb R^n$.
	Then the following statements hold.
	\begin{enumerate}
		\item[(i)] \eqref{eq:all_ones_tensor_eq} has a positive solution $x\in\mathbb R_{++}^n$ if and only if
		\begin{equation}\label{eq:b_const_vector}
			b=c\mathbf 1,\qquad c>0,
		\end{equation}
		where $\mathbf 1=(1,\dots,1)^\top\in\mathbb R^n$.
		In this case, every positive solution satisfies
		\begin{equation}\label{eq:sum_constraint}
			\sum_{j=1}^n x_j = c^{\frac{1}{k-1}}.
		\end{equation}
		
		\item[(ii)] \eqref{eq:all_ones_tensor_eq} has a \emph{unique} positive solution if and only if $n=1$ and
		$b=c$ with scalar $c>0$. In that case, the unique positive solution is
		\[
		x=c^{\frac{1}{k-1}}.
		\]
		Moreover, if $n\ge 2$ and $b=c\mathbf{1}$ with $c>0$, then \eqref{eq:all_ones_tensor_eq} admits infinitely many
		positive solutions, namely all $x\in\mathbb R_{++}^n$ satisfying \eqref{eq:sum_constraint}.
	\end{enumerate}
\end{proposition}

\begin{proof}
	For each $i\in\{1,\dots,n\}$, using $a_{i_1\cdots i_k}=1$ and the standard tensor-vector product,
	$(\mathcal A x^{k-1})_i
	=
	\sum_{i_2,\dots,i_k=1}^n a_{i\,i_2\cdots i_k}\,x_{i_2}\cdots x_{i_k}
	=
	\sum_{i_2,\dots,i_k=1}^n x_{i_2}\cdots x_{i_k}
	=
	\Big(\sum_{j=1}^n x_j\Big)^{k-1}.$
	Hence $\mathcal A x^{k-1}=s^{k-1}\mathbf 1$ with $s=\sum_{j=1}^n x_j$.
	Therefore \eqref{eq:all_ones_tensor_eq} is equivalent to $b=s^{k-1}\mathbf 1$.
    
	(i)  If there exists $x\in\mathbb R_{++}^n$, then $s>0$ and thus $b=c\mathbf 1$ with $c=s^{k-1}>0$.
	Conversely, if $b=c\mathbf 1$ with $c>0$, let $s=c^{1/(k-1)}$ and pick any $x\in\mathbb R_{++}^n$ such that
	$\sum_{j=1}^n x_j=s$; then $\mathcal A x^{k-1}=s^{k-1}\mathbf 1=c\mathbf 1=b$.
	
	  (ii)  If $n=1$, then $s=x$ and the equation reduces to $x^{k-1}=c$, which has the unique positive solution
	$x=c^{1/(k-1)}$. If $n\ge 2$ and $b=c\mathbf 1$ with $c>0$, then the positive solution set is exactly
	$\{x\in\mathbb R_{++}^n:\sum_{j=1}^n x_j=c^{1/(k-1)}\}$, which is infinite, so uniqueness fails.
    \end{proof}

\begin{proposition}\label{prop:ones_spos_noJ}
	Let $m\ge 2$ and $p=m-1$. Let $\mathcal E$ be the all-ones $m$th-order $n$-dimensional tensor and let $I$ be the
	identity tensor. Fix $a>0$, $s>0$, and $b\in\mathbb{R}_{++}^n$. Let
	\[
	\tau=a\Big(\sum_{j=1}^n x_j\Big)^{p}.
	\]
    Consider
	\begin{equation}\label{eq:ones_spos_noJ}
		a\,\mathcal E\,x^{p}+s\,I x^{p}=b,\qquad x\in\mathbb{R}_{++}^n.
	\end{equation}
	Define $b_{\min}=\min_i b_i$ and $b_{\max}=\max_i b_i$, and
	\[
	H(\tau)=\frac{a}{s}\Bigg(\sum_{i=1}^n(b_i-\tau)^{1/p}\Bigg)^{p},\qquad \tau\in(0,b_{\min}).
	\]
	Then:
	\begin{enumerate}
	\item[(i)] Equation~\eqref{eq:ones_spos_noJ} has a solution in $\mathbb{R}_{++}^n$ if and only if there exists
	$\tau\in(0,b_{\min})$ such that
	\begin{equation}\label{eq:tau_fp_spos_noJ}
		\tau=H(\tau).
	\end{equation}
	In that case the (positive) solution is given by
	\begin{equation}\label{eq:x_from_tau_spos_noJ}
		x_i=\Big(\frac{b_i-\tau}{s}\Big)^{1/p},\qquad i=1,\dots,n.
	\end{equation}
    \item[(ii)] 	If a positive solution exists, it is unique.
    \item[(iii)]   Equation~\eqref{eq:ones_spos_noJ} has a positive solution if and only if
	\begin{equation}\label{eq:exist_iff_spos_noJ}
		H(b_{\min}^-) \le b_{\min}.
	\end{equation}
	A convenient sufficient condition for \eqref{eq:exist_iff_spos_noJ} is
	\begin{equation}\label{eq:exist_suff_bminbmax_noJ}
		\frac{a\,n^{p}}{s}\,(b_{\max}-b_{\min})<b_{\min}.
	\end{equation}
    \item[(iv)] For the (unique) $\tau$ solving \eqref{eq:tau_fp_spos_noJ}, one has
	\begin{equation}\label{eq:tau_bounds_spos_noJ}
		\frac{a\,n^{p}\,b_{\min}}{s+a\,n^{p}}
		\;\le\;
		\tau
		\;\le\;
		\frac{a\,n^{p}\,b_{\max}}{s+a\,n^{p}}.
	\end{equation}
	\end{enumerate}
\end{proposition}

\begin{proof}
	Since $\mathcal E$ is the all-ones tensor, for every $i$ one has
	\[
	(\mathcal E x^{p})_i=\sum_{i_2,\dots,i_m=1}^n x_{i_2}\cdots x_{i_m}
	=\Big(\sum_{j=1}^n x_j\Big)^{p}.
	\]
	Then \eqref{eq:ones_spos_noJ} is equivalent to
	\[
	\tau+s x_i^{p}=b_i,\qquad i=1,\dots,n,
	\]
	so necessarily $0<\tau<b_{\min}$ and
	\[
	x_i=\Big(\frac{b_i-\tau}{s}\Big)^{1/p},\qquad i=1,\dots,n.
	\]
	Summing over $i$ and using the definition of $\tau$ yields the fixed-point equation \eqref{eq:tau_fp_spos_noJ}.
	
	For $\tau\in(0,b_{\min})$, define
	$S(\tau)=\sum_{i=1}^n(b_i-\tau)^{1/p},
	H(\tau)=\frac{a}{s}\,S(\tau)^{p},
	F(\tau)=H(\tau)-\tau.$
	Each term $(b_i-\tau)^{1/p}$ is strictly decreasing in $\tau$, hence $S(\tau)$ is strictly decreasing, and so is $H(\tau)$.
	Therefore, for $0<\tau_1<\tau_2<b_{\min}$,
	\[
	F(\tau_2)-F(\tau_1)=(H(\tau_2)-H(\tau_1))-(\tau_2-\tau_1)<0,
	\]
	so $F$ is strictly decreasing and has at most one zero, proving uniqueness.
	
	Moreover, $F(0)=H(0)>0$. Hence a zero exists in $(0,b_{\min})$ if and only if $F(b_{\min}^-)\le 0$, i.e.,
	$H(b_{\min}^-)\le b_{\min}$, which is \eqref{eq:exist_iff_spos_noJ}. Using $b_i-b_{\min}\le b_{\max}-b_{\min}$ gives
	$H(b_{\min}^-)
	=\frac{a}{s}\Bigg(\sum_{i=1}^n(b_i-b_{\min})^{1/p}\Bigg)^{p}
	\le \frac{a}{s}\big(n(b_{\max}-b_{\min})^{1/p}\big)^{p}
	=\frac{a\,n^{p}}{s}(b_{\max}-b_{\min}),$
	so \eqref{eq:exist_suff_bminbmax_noJ} implies \eqref{eq:exist_iff_spos_noJ}.
	
	Finally, let $\tau$ be the (unique) fixed point. From \eqref{eq:tau_fp_spos_noJ} and $b_i\ge b_{\min}$,
	$\tau=\frac{a}{s}\Bigg(\sum_{i=1}^n(b_i-\tau)^{1/p}\Bigg)^{p}
	\ge \frac{a}{s}\big(n(b_{\min}-\tau)^{1/p}\big)^{p}
	=\frac{a\,n^{p}}{s}(b_{\min}-\tau),$
	which gives $\tau\ge \frac{a\,n^{p}\,b_{\min}}{s+a\,n^{p}}$. Similarly, using $b_i\le b_{\max}$ gives
	\[
	\tau\le \frac{a}{s}\big(n(b_{\max}-\tau)^{1/p}\big)^{p}
	=\frac{a\,n^{p}}{s}(b_{\max}-\tau),
	\]
	hence $\tau\le \frac{a\,n^{p}\,b_{\max}}{s+a\,n^{p}}$.
	This proves \eqref{eq:tau_bounds_spos_noJ}.
\end{proof}

	\section{The existence of equilibrium points in High-order LV model.}\label{existence}
	
	To enable precise analysis of network competition, scholars derive a LV equation for firing activity from the membrane dynamics equation. In \cite{Cowan68}-\cite{Cowan70}, Cowan discussed a LV dynamics. Within the canonical framework of neural dynamics, the membrane potential and firing rate of the $i$-th neuron are mathematically represented as follows:
	\begin{equation}\label{model'} 
		\begin{aligned}\frac{du_i}{dt}&=-\lambda u_i+(\text{input current}),\\
			z_i&= f(u_i-\theta) ,\end{aligned}
	\end{equation}
	where $u_{i}$ and $z_{i}$ represent membrane potential and firing rate of neuron $i$, respectively. $f(u_i -\theta)$ denotes a nonlinear output function.
	
	On the basis of the model (\ref{model'}), we consider an $n$-neuron system under uniform 3-body interactions networks, the firing rate dynamics can be formulated as:
	\begin{equation}\label{model-1} 
		\frac{dz_{i}}{dt}=z_{i}(1-z_{i}^{2}-k\sum_{(p,q)\neq(i,i)}z_{p}z_{q}+w_{i}),~~i=1,2,\cdots,n,
	\end{equation}
	where $z_{i}$ represents firing rate of neuron $i$. Neurons $p$ and $q$ are neighbors connected to $i$ through hyperedges. $k$ denotes the ratio of self-inhibition to lateral-inhibition. $1$ represents the input current induced by the afferent inputs nonspecific to each neuron, and $w_i>0$ represents the input current generated by the specific afferent inputs. 
	
	Rewrite equation (\ref{model-1}) into tensor form as follows:
	
	\begin{equation}\label{model-1-tensor} 
		\dot{z}=diag(z)(b+Az^{2}),
	\end{equation}
	where $z=(z_{1}, z_{2}, \cdots, z_{n})^{T}$, $z^2=z\otimes z$, $b=\begin{pmatrix}1+w_{1}\\1+w_{2}\\\vdots\\1+w_{n}\end{pmatrix}$,  \\
	$\begin{aligned}
		&A_{(i,:,:)}=\begin{array}{@{}r@{}c@{}c@{}c@{}c@{}c@{}c@{}l@{}}
			& 1 & 2 & \cdots & i &\cdots & n \\
			\left.\begin{array}
				{c} 1 \\2 \\\vdots \\i \\\vdots\\n \end{array}\right(
			& \begin{array}{c} -k \\ -k             \\ \vdots \\ -k       \\ \vdots     \\ -k               \end{array}
			& \begin{array}{c} -k          \\ -k  \\ \vdots \\ -k      \\ \vdots     \\-k                 \end{array}
			& \begin{array}{c} \cdots      \\ \cdots       \\ \ddots \\ \cdots  \\    \ddots        \\ \cdots          \end{array}
			& \begin{array}{c} -k           \\ -k             \\ \vdots \\ -1 \\ \vdots     \\-k                \end{array}
			& \begin{array}{c} \cdots      \\ \cdots        \\  \ddots       \\  \cdots  \\ \ddots     \\ \cdots           \end{array}
			& \begin{array}{c} -k           \\  -k            \\ \vdots \\ -k       \\ \vdots     \\-k      \end{array}
			& \left).\begin{array}{c} \\ \\ \\ \\ \\ \\ \end{array}\right.
		\end{array}
	\end{aligned}$
	
	\begin{definition}\label{winner-d}
		Without loss of generality, we can assume that $\mathcal{D}$ is the set of winners produced by the system (\ref{model-1}), where $ d\in [0, n]$ represents the number of winners. 
	\end{definition}
	
	Next, we calculate equilibrium points for System (\ref{model-1}). By searching for the equilibrium points, we can roughly determine the relationship between the number of winners 
	$d$ and the ratio of self-inhibition to lateral-inhibition $k$ in System (\ref{model-1}).
	
	\begin{theorem}\label{theorem-1-equilibrium-points}
		For system (\ref{model-1}), when $k>1$, there can only be one winner, i.e. $d=1$. When $k<1$, multiple winners may exist, i.e. $d\geq 1$. When $k=1$, there can only be one winner, if and only if there exists no $i,j$ such that $w_i=w_j$. 
	\end{theorem}
	
	\begin{proof}
		The equilibrium points of the system (\ref{model-1}) can be obtained according to the following equation:
		\[\dot{z}=diag(z)(b+Az^{2})=0,\]
		
		We can obtain the equilibrium point $z_{e_1}^*=(0, 0, \cdots, 0)^{T}$, which corresponds to the outcome where all neurons are losers and die out ($d=0$), or the equilibrium point satisfying the equation $-Az^{2}=b$, corresponding to the outcome where winners exist ($d\geq 1$). 
		
		When $z=z_{e_1}^*=(0, 0, \cdots, 0)^{T}$, it is easy to get:
		\[
		\begin{aligned}
			\frac{\partial\dot{z}_i}{\partial z_i}\Big|_{z=z_{e_1}^*}
			&=(1-z_i^2 - k\!\!\sum_{p,q\neq(i,i)} z_p z_q + w_i) \\
			&\quad + z_i\!\left(
			-k \!\!\sum_{p=i,\,q\neq i} z_q
			-k \!\!\sum_{q=i,\,p\neq i} z_p
			- 2z_i
			\right)\Big|_{z=z_{e_1}^*} \\
			&= 1+w_i > 0, \\
			\frac{\partial\dot{z}_i}{\partial z_j}\Big|_{z=z_{e_1}^*}
			&= z_i\!\left(
			-k \!\!\sum_{q=1}^{n} z_q
			-k \!\!\sum_{p=1}^{n} z_p
			\right)\Big|_{z=z_{e_1}^*} = 0.
		\end{aligned}
		\] 
		It is known that the diagonal elements of a diagonal matrix are its eigenvalues. Clearly, the eigenvalues of the Jacobian matrix for system (\ref{model-1}) are positive. 
        It is easy to conclude that the equilibrium point $z_{e_1}^*=(0, 0, \cdots, 0)^{T}$ is unstable. Therefore, we know that the neurons in the system (\ref{model-1}) will inevitably produce at least one winner through competition. That is to say $d \neq 0$.
		
		When $d\geq 1$, the set of equilibrium points $\{z_i^{(\mathcal{D})}\}$ satisfies the following:
		\begin{equation}\label{model-d}    
			\left\{
			\begin{aligned}
				&\sum_{p,q \in \mathcal{D}}A_{ipq}^{(\mathcal{D})}z_{p}z_{q}=1+w_{i},&&~~i\in\mathcal{D},\\
				&z_{i}^{(\mathcal{D})}=0,&&~~i\notin \mathcal{D},
			\end{aligned}
			\right.
		\end{equation}
		where 
\[
A_{(i,:,:)}^{(\mathcal D)}=
\begin{blockarray}{c@{\hspace{0.9em}}cccccc}
	& 1 & 2 & \cdots & i & \cdots & d \\
	\begin{block}{c@{\hspace{0.9em}}(cccccc)}
		1      & -k & -k & \cdots & -k & \cdots & -k \\
		2      & -k & -k & \cdots & -k & \cdots & -k \\
		\vdots & \vdots & \vdots & \ddots & \vdots & \ddots & \vdots \\
		i      & -k & -k & \cdots & -1 & \cdots & -k \\
		\vdots & \vdots & \vdots & \ddots & \vdots & \ddots & \vdots \\
		d      & -k & -k & \cdots & -k & \cdots & -k \\
	\end{block}
\end{blockarray}.
\]
		For $i\in\mathcal{D}$, we have \[-A^{(\mathcal{D})}z^{(\mathcal{D})2}=b^{(\mathcal{D})}.\] 
		Here, $z^{(\mathcal{D})}$ and $b^{(\mathcal{D})}$ represent the vectors $z$ and $b$ with  with the elements $z_i (i\notin \mathcal{D})$ and $b_i$ $(i\notin \mathcal{D})$ removed, respectively. 
		
		We already know that the diagonal elements of $z^{(\mathcal{D})}$ are positive. According to Lemma \ref{unique-positive-solution-1}, when $|A_{iii}|=1 \textgreater \sum_{(j,k)\neq(i,i)}|A_{ijk}|=k(d^2-1)$, the equation $-A^{(\mathcal{D})}z^{(\mathcal{D})2}=b^{(\mathcal{D})}$ has a unique positive solution.  Through simple calculations, we find that if $d<\sqrt{\frac{k+1}k}$ holds, the equation (\ref{model-d})  has unique solution $z_{e_2}^*$, where $z_{e_2i}^*>0$ (for $i \in \mathcal{D}$) and $z_{e_2i}^*=0$ (for $i \notin \mathcal{D}$). Additionally, when $k\geq1$, there must exists at least one winner, i.e. $d=1$. When $k<1$, multiple winners may exist, i.e. $d\geq 1$. 
		
		Next, we show that when $k>1$, any equilibrium point with more than one winner ($d>1$) is unstable, if exists.
		Fix an equilibrium point associated with a winner set $\mathcal D$ with $|\mathcal D|=d>1$.
		For convenience, denote $b_i=1+w_i$ and write the dynamics of~(\ref{model-1}) componentwise as
		\[
		\dot z_i
		=
		z_i\!\left(
		b_i - z_i^2 - k\!\!\sum_{p,q\neq(i,i)} z_p z_q
		\right),\qquad i=1,\dots,n .
		\]
		Note that
		\[
		\sum_{p,q=1}^n z_p z_q=\Big(\sum_{\ell=1}^n z_\ell\Big)^2=:S(z),
		\quad
		\sum_{p,q\neq(i,i)} z_p z_q = S(z)-z_i^2,
		\]
		so that
		\begin{equation}\label{eq:rewrite}
			\dot z_i
			=
			z_i\Big(b_i - kS(z) + (k-1)z_i^2\Big).
		\end{equation}
		Let $z^*$ be an equilibrium point with winners $\mathcal D$.
		Then $z_i^*>0$ for $i\in\mathcal D$ and $z_i^*=0$ for $i\notin\mathcal D$.
		Define
		\[
		r=\sum_{i\in\mathcal D} z_i^*>0,
		\qquad
		S^*=S(z^*)=r^2.
		\]
		Moreover, for each $i\in\mathcal D$, the equilibrium condition $\dot z_i=0$ together with~\eqref{eq:rewrite} yields
		\begin{equation}\label{eq:equil_winner}
			b_i - kS^* + (k-1)(z_i^*)^2 = 0, \qquad i\in\mathcal D.
		\end{equation}
		
		We now compute the Jacobian matrix $J(z)=\left[\frac{\partial \dot z_i}{\partial z_j}\right]$ at $z=z^*$.
		For $i\in\mathcal D$, differentiating~\eqref{eq:rewrite} gives
		\[
		\frac{\partial \dot z_i}{\partial z_j}(z^*)
		=
		z_i^*\frac{\partial}{\partial z_j}\Big(b_i-kS(z)+(k-1)z_i^2\Big)\Big|_{z=z^*},
		\]
		where we used $b_i-kS^*+(k-1)(z_i^*)^2=0$ from~\eqref{eq:equil_winner}.
		Since $S(z)=(\sum_\ell z_\ell)^2$, we have
		\[
		\frac{\partial S}{\partial z_j}(z^*) = 2\sum_{\ell=1}^n z_\ell^* = 2r.
		\]
		Therefore, for $i\in\mathcal D$,
		\[
		\frac{\partial \dot z_i}{\partial z_j}(z^*)
		=
		\begin{cases}
			-2kr\,z_i^*, & j\in\mathcal D,\ j\neq i,\\[2mm]
			-2kr\,z_i^*+2(k-1)(z_i^*)^2, & j=i\in\mathcal D.
		\end{cases}
		\]
		Hence, the winner--winner block $J_{\mathcal D}(z^*)\in\mathbb R^{d\times d}$ can be written as
		\begin{equation}\label{eq:Jblock}
			J_{\mathcal D}(z^*)
			=
			u\mathbf 1^\top + 2(k-1)\operatorname{diag}\big((z_i^*)^2\big),
			\qquad
			u_i=-2kr\,z_i^*,
		\end{equation}
		where $\mathbf 1\in\mathbb R^d$ is the all-ones vector.
		
		Since $d>1$, there exists a nonzero vector $v\in\mathbb R^d$ such that $\mathbf 1^\top v=0$.
		Multiplying~\eqref{eq:Jblock} by such a $v$ yields
		$v^\top J_{\mathcal D}(z^*) v
		=
		v^\top\!\left(u\mathbf 1^\top\right)\!v
		+
		2(k-1)\,v^\top \operatorname{diag}\big((z_i^*)^2\big)v
		=
		0 + 2(k-1)\sum_{i\in\mathcal D}(z_i^*)^2 v_i^2 .$
		
		When $k>1$, we have $k-1>0$ and $z_i^*>0$ for all $i\in\mathcal D$, hence
		$v^\top J_{\mathcal D}(z^*) v=v^\top \frac{J_{\mathcal D}(z^*)+J_{\mathcal D}(z^*)^\top}{2} v > 0.$
		By the Rayleigh quotient, this implies that $J_{\mathcal D}(z^*)$ has an eigenvalue with a positive real part.
		Since $J_{\mathcal D}(z^*)$ is a principal submatrix of the full Jacobian $J(z^*)$, the Jacobian $J(z^*)$ also has an eigenvalue with a positive real part.
		Therefore, 
        the equilibrium point $z^*$ is unstable.
		
		When $k=1$, note that $A^{(\mathcal{D})}$ is an all-one tensor, thus it directly follows from the proposition \ref{prop:all_ones_tensor_unique_pos}.
		
		Consequently, when $k>1$, any equilibrium point with more than one winner ($d>1$) is unstable.
		This shows that for $k>1$, the only asymptotically stable competitive outcome must satisfy $d=1$.
	\end{proof}
	
	\begin{remark}\label{remark-1}
		It is evident that the smaller the value of $k$ is, the greater the number of winners will be. This indicates that the coexistence of multiple neurons could be more promising at the small $k$.
	\end{remark}

	\section{The stability of all-neurons-coexistence equilibrium point.}\label{stability-all-neurons-coexistence}
	
	In this section, we analyze the stability of the all-neurons-coexistence equilibrium point, where there is no loser. We further show both its sufficient and necessary condition for the local asymptotic stability and its sufficient condition for the global asymptotic stability.
	
	\begin{theorem}\label{theorem-2-all-neurons-coexistence-equilibrium-points}
		Assuming that system (\ref{model-1}) has an all-neurons-coexistence equilibrium $z_{e_3}^*=(z_{1}^*, z_{2}^*, \cdots, z_{n}^*)^T$, where $z_{i}^*>0$ for all neuron $i$. The equilibrium point  $z_{e_3}^*$ is locally asymptotically stable if and only if $k<1$.
	\end{theorem}

	\begin{proof}
		For system (\ref{model-1}), it is easy to obtain that: 
		\[
		\begin{aligned}
			\frac{\partial\dot{z}_i}{\partial z_i}\Big|_{z=z_{e_3}^*}
			&= (1 - z_i^2 - k\!\!\sum_{p,q\neq(i,i)} z_p z_q + w_i) \\
			&\quad + z_i\!\left(
			-k\!\!\sum_{p=i,\,q\neq i} z_q
			-k\!\!\sum_{q=i,\,p\neq i} z_p
			- 2z_i
			\right)\Big|_{z=z_{e_3}^*} \\
			&= z_{i}^{*}\!\left(
			-k\!\!\sum_{q=1}^{n} z_{q}^{*}
			-k\!\!\sum_{p=1}^{n} z_{p}^{*}
			+ 2k z_{i}^{*}
			- 2 z_{i}^{*}
			\right), 
		\end{aligned}\]
			\[\begin{aligned}
			\frac{\partial\dot{z}_i}{\partial z_j}\Big|_{z=z_{e_3}^*}
			&= z_{i}\!\left(
			-k\!\!\sum_{q=1}^{n} z_{q}
			-k\!\!\sum_{p=1}^{n} z_{p}
			\right)\Big|_{z=z_{e_3}^*} \\
			&= z_i^*\!\left(
			-k\!\!\sum_{q=1}^n z_q^*
			-k\!\!\sum_{p=1}^n z_p^*
			\right) < 0\qquad (i\neq j).
		\end{aligned}
		\]
		The Jacobian matrix $J$ of the system (\ref{model-1}) is given by:
		\begin{equation*}
			J=\hat{T}+\hat{Z},
		\end{equation*}
		where {\tiny $\hat{T}=\left(\begin{matrix}T_1&T_1&\cdots&T_1\\T_2&T_2&\cdots&T_2\\\vdots&\vdots&\ddots&\vdots\\T_n&T_n&\cdots&T_n\end{matrix}\right)_n$, $\hat{Z}=2(k-1)\left(\begin{matrix}z_{1}^{*^{2}}& 0 & \cdots & 0\\ 0 &z_{2}^{*^{2}}& \cdots & 0\\\vdots&\vdots &\ddots& \vdots\\ 0& 0& \cdots &z_{n}^{*^{2}}\end{matrix}\right)$,}
		$T_i=-z_i^*\big(k\sum_{q=1}^nz_q^*+k\sum_{p=1}^nz_p^*\big)$.
		Denote $r=\sum_{\ell=1}^{n}z_{\ell}^{*}>0$. Then $T_i=-2kr\,z_i^*$ and $\hat{T}=t\mathbf 1^{\top}$ with $t=(T_1,\dots,T_n)^{\top}$.
		
		(i) When $k=1$, we have $J=\hat{T}=t\mathbf 1^{\top}$, which is a rank-one matrix. Hence, $J$ has one nonzero eigenvalue
		$\lambda=-2r^2<0,$
		and the remaining $n-1$ eigenvalues are $0$. Therefore, $J$ is not Hurwitz, and the equilibrium point $z_{e_3}^*$ is not locally asymptotically stable.
		
		(ii) When $k>1$, we have $k-1>0$ and thus $\hat{Z}$ is a positive diagonal matrix. Since $n>1$, there exists a nonzero vector $v\in\mathbb R^{n}$ such that $\mathbf 1^{\top}v=0$.
		For such a $v$, we have $\hat{T}v=t(\mathbf 1^{\top}v)=0$ and thus
		\[
		v^{\top}Jv=v^{\top}\hat{Z}v=2(k-1)\sum_{i=1}^{n}z_{i}^{*^{2}}v_i^2>0.
		\]
		Consider the symmetric part $J_s=\frac{1}{2}(J+J^{\top})$. Since $v^{\top}J_sv=v^{\top}Jv>0$, we have $\lambda_{\max}(J_s)>0$, which implies that $J$ has an eigenvalue with positive real part. Therefore, 
         the equilibrium point $z_{e_3}^*$ is unstable when $k>1$.
		
		(iii) When $k<1$, we have $k-1<0$ and thus $\hat{Z}$ is a negative diagonal matrix. Define a positive diagonal matrix
		\[
		P=\mathrm{diag}\Big(\frac{1}{z_1^*},\frac{1}{z_2^*},\dots,\frac{1}{z_n^*}\Big)\succ 0.
		\]
		Using $\hat{T}=t\mathbf 1^{\top}$ and $t_i=-2kr\,z_i^*$, we obtain
		\[
		P\hat{T}=-2kr\,\mathbf 1\mathbf 1^{\top},\qquad
		P\hat{Z}=2(k-1)\,\mathrm{diag}(z_1^*,\dots,z_n^*).
		\]
		Consequently,
		$	PJ+J^{\top}P
		=
		(P\hat{T}+\hat{T}^{\top}P)+(P\hat{Z}+\hat{Z}^{\top}P)
		=
		-4kr\,\mathbf 1\mathbf 1^{\top}+4(k-1)\,\mathrm{diag}(z_1^*,\dots,z_n^*).$
		
		Since $k<1$ and $z_i^*>0$ for all $i$, the second term is negative definite and the first term is negative semidefinite, hence
		$PJ+J^{\top}P\prec 0.$
		Therefore, $J$ is Hurwitz and the all-neurons-coexistence equilibrium point $z_{e_3}^*$ is locally asymptotically stable.
		The necessary condition is a direct consequence of Theorem \ref{theorem-1-equilibrium-points}.
	\end{proof}
	
	\begin{remark}\label{remark-existence-coexistence}
		The local asymptotic stability result for the all-neurons-coexistence equilibrium point is conditional on its
		\emph{existence}. Namely, the condition $k<1$ guarantees that \emph{if} a coexistence equilibrium
		$z^*\gg 0$ exists, then it is locally asymptotically stable.
		However, $k<1$ alone does not automatically ensure the existence of such a strictly positive equilibrium;
		additional structural assumptions 
		are required to guarantee existence.
	\end{remark}

	\begin{theorem}\label{theorem-3-all-neurons-coexistence-equilibrium-points}
		Let that the unique positive equilibrium point of system (\ref{model-1-tensor}) is $z_{e_3}^*$. If $-A$ is an irreducible nonnegative $\mathcal{H}^+$-tensor, the equilibrium point $z_{e_3}^*$ exists and is globally asymptotically stable.
	\end{theorem}
	\begin{proof}
		In the equation (\ref{model-1-tensor}), if $-A$ is irreducible nonnegative $\mathcal{H}^+$-tensor,  according to Lemma \ref{unique-positive-solution-2}, we know that the system (\ref{model-1-tensor}) has unique positive equilibrium point $z_{e_3}^*$. Define set $U_1=\{z|z\geq z_{e_3}^*\}$ and $U_2=\{z|\min \frac{z_i}{z_{e_3i}^*}<1\}$. Then we let $V_m= \min_i (\frac{z_i}{z^*_{e_3i}})^2$ and $V= \min_i (\frac{z_i-z_{e_3i}^*}{z^*_{e_3i}})^2$, where $m=\arg\min_i (\frac{z_i}{z^*_{e_3i}})^2$ and assuming $t=\arg\min_i (\frac{z_i-z^*_{e_3i}}{z^*_{e_3i}})^2$. Because $\frac{z_i-z^*_{e_3i}}{z^*_{e_3i}}=\frac{z_i}{z^*_{e_3i}}-1$, $m=t$ must hold. In the set $U_1$, we have:
		\begin{equation*}
			\begin{aligned}
				\dot{V}& =\frac{2z_{m}}{(z_{m}^{*})^{2}} (z_{m}-z_{m}^{*})(b+AZ^{2})_{m}   \\
				& =\frac{2z_{m}}{(z_{m}^{*})^{2}} (z_{m}-z_{m}^{*})(1+w_m+\sum_{jk}a_{mjk} \frac{z_jz_kz_j^*z_k^*}{z_j^*z_k^*})\\
				&\leq\frac{2z_m}{(z_m^*)^2}(z_m-z_m^*)\big(1+w_m+\sum_{jk}a_{mjk} z_j^*z_k^*V_m\big)\\
				&=\frac{2z_m}{(z_m^*)^2}(z_m-z_m^*)(1+w_m-(1+w_m)V_m) \\
				&=\frac{2z_{m}}{(z_{m}^{*})^{2}}(z_{m}-z_{m}^{*})(1-V_{m})(1+w_{m})\leq0.\\ 
			\end{aligned}
		\end{equation*}
		For any $z\in U_2$, we can obtain that:
		\begin{equation*}
			\begin{aligned}
				\dot{V}& =\frac{2z_{m}}{(z_{m}^{*})^{2}} (z_{m}-z_{m}^{*})(b+AZ^{2})_{m}            \\
				&\leq\frac{2z_m}{(z_m^*)^2}(z_m-z_m^*)\left(1+w_m+\sum _{jk}a_{mjk}z_j^*z_k^*V_m^\frac{1}{2}\right) \\ 
				&=\frac{2z_m}{(z_m^*)^2}(z_m-z_m^*)(1+w_m)(1-V_m^\frac{1}{2})\leq0.
			\end{aligned}
		\end{equation*}
		
		Notice that $V$ is locally positive definite in $U_1$ and $U_2$, and $\dot{V}$ is negative semidefinite. Therefore, the solution converges to $\{x|\dot{V}=0\}$. In order to have $\dot{V}=0$, for every $i$, it must hold that $i=m$ and $z_m=z_m^*$ which indicates that $\{z|\dot{V}=0\}$ only contains $z_{e_3}^*$. Thus, $z_{e_3}^*$ is asymptotically stable with domain of attraction $U_1, U_2$.
		
		Next, we further let $L_s=\max_i(\frac{z_i}{z^*_{e_3i}})^2$ and $s=\arg \max_i(\frac{z_i}{z^*_{e_3i}})^2$.
		Similar to previous arguments we have that $-z_i=-\frac{z_iz_{e_3i}^*}{z_{e_3i}^*}\geq -L_s^{\frac{1}{2}}z_{e_3i}^*$, which holds as a equality when $i=s$ and otherwise holds as an strict inequality. Then, consider the set $U_3=\{z| z<z_{e_3}^*\}$, $\forall z\in U_{3}$, we get:
		\begin{equation*}
			\begin{aligned}
				\dot{V}&\geq\frac{2z_{s}^{2}}{(z_{s}^{*})^{2}}\big(1+w_{s}+\sum _{jk}a_{sjk} z_{j}^{*}z_{k}^{*} Ls\big)\quad\\
				&=\frac{2z_{s}^{2}}{(z_{s}^{*})^{2}}(1+w_{s})(1-Ls)>0,\quad
			\end{aligned}
		\end{equation*}
		This guarantees that all trajectories in $U_3$ finally enter $U_1\cup U_2$. Thus, the equilibrium point $z_{e_3}^*$ is globally asymptotically stable.	
	\end{proof}
	
	\begin{remark}
		Recall that a strictly diagonally dominant tensor with all diagonal entries is an $\mathcal{H}^{+}$-tensor. Therefore, the all-neuron coexistence is preferred when $k$ is small.
	\end{remark}

	\section{The stability of equilibrium points where losers exist.}\label{stability-winner-share-all}
	This section primarily discusses the competitive evolution results of System (\ref{model-1}) for non-all-neurons-coexistence equilibrium points, where losers are present.
	
	\begin{definition}(\cite{Fukai-97})
		With respect to an initial condition, after competitive evolution, when only the neuron with the highest initial input becomes the winner, it is referred to as the WTA phenomenon. If only one winner is activated, but this winner is not necessarily the neuron with the highest input, this behavior is called the VWTA. When the system ultimately produces multiple winners, it is termed WSA. 
	\end{definition}

	\begin{theorem}\label{three-result}
		When $k=1$, system (\ref{model-1}) evolves through competition into a WTA phenomenon, with the sole winner having the highest initial input $W_{max}$. When $k>1$, system (\ref{model-1}) exhibits VWTA behavior after competition. When $k<1$, the system (\ref{model-1})'s competitive outcome is either WSA or WTA.
	\end{theorem}
	
	\begin{proof}
		For the case where the number of winners $1\leq d<n$, with a simple renumbering of the winners, we can set the equilibrium point as  $z_{e_4}^*=(z_{1}^*, z_{2}^*, \cdots, z_{d}^*, 0, \cdots, 0)^T$. At this point, the Jacobian matrix of system (\ref{model-1}) is denoted as $J'=\begin{pmatrix}J_{winner}&\Omega\\0_{(n-d) \times d}&Res\end{pmatrix}$. where
		$
		J_{\mathrm{winner}}=
		\begin{pmatrix}
			\alpha_1 & T_1 & \cdots & T_1\\
			T_2 & \alpha_2 & \cdots & T_2\\
			\vdots & \vdots & \ddots & \vdots\\
			T_d & T_d & \cdots & \alpha_d
		\end{pmatrix}$,  
		 $\Omega=\left(\begin{matrix}T_1&T_1&\cdots&T_1\\T_2&T_2&\cdots&T_2\\\vdots&\vdots&\ddots&\vdots\\T_d&T_d&\cdots&T_d\end{matrix}\right)$, and $\alpha_i= T_i + 2(k-1)z_i^{*2}~(i=1,\ldots,d)$. The eigenvalues of matrix $J'$ consist of the eigenvalues of $J_{winner}$ and the eigenvalues of 
		$Res$.
		
		Without loss of generality, we assume $i$ is the loser and $j$ is the winner. Then, for the diagonal elements of 
		$Res$, we have:
		\begin{equation}\label{RES}
			\frac{\partial\dot{z}_i}{\partial z_i}=1+w_i-k\sum_{ p,q \in \mathcal{D}}z_p^*z_q^*.
		\end{equation}
		For the off-diagonal elements of $Res$, one obtain:
		\begin{equation*}
			\frac{\partial\dot{z}_l}{\partial z_j}=0,\quad for~ j\neq i.
		\end{equation*}
		
		According to Equation (\ref{model-1}), it is straightforward to deduce that for the winner $j$, the  equality $1+w_j-\left(z_j^*\right)^2-k\sum_{p,q \in \mathcal{D}, p\neq q}z_p^*z_q^*=0$	holds, i.e. 
		$1+w_j-\left(z_j^*\right)^2-k\sum_{p,q \in \mathcal{D}}z_p^*z_q^*+k\left(z_j^*\right)^2=0$.
		Substituting $k\sum_{p,q\in\mathcal{D}}z_p^*z_q^*=1+w_j+(k-1)\Big(z_j^*\Big)^2$ into Equation (\ref{RES}), we obtain the following equality: $\frac{\partial\dot{z}_i}{\partial z_i}=w_i-w_j-(k-1)\left(z_j^*\right)^2$. 
		
		(i) When $k=1$, according to Theorem \ref{theorem-1-equilibrium-points}, we have $d=1$. Moreover, similar to the proof of Theorem \ref{theorem-2-all-neurons-coexistence-equilibrium-points}, we conclude that the eigenvalues of $J_{winner}$ are non-positive. Thus, the eigenvalues of the matrix $Res$ are negative if and only if $w_{j}=w_{max}$, indicating that the equilibrium point is locally lyapunov stable. This indicates that when $k=1$, the WTA equilibrium point of system (\ref{model-1}) is stable under small perturbations, where the neuron with the maximum initial input emerges as the winner. However, without further analysis, the system may not necessarily converge to this state. The detailed asymptotic stability analysis will be given in the next section.
		
		(ii) When $k>1$, Theorem \ref{theorem-1-equilibrium-points} also indicates that there is only one winner in this case. At this point, $J_{winner}=-2kz_{j}^{*2}<0$. It is worth noting that if $k$ is sufficiently large, $w_j=w_{max}$ is not required for the eigenvalues of the matrix $Res$ to be negative. In other words, when $k>1$, the system (\ref{model-1}) will achieve VWTA.
		
		(iii) When $k<1$, since the similar proof as in the Theorem \ref{theorem-2-all-neurons-coexistence-equilibrium-points} applies to the winner subsystem, we have $J_{\mathrm{winner}}$ is Hurwitz.
		
		Next, denote $r=\sum_{p\in D}z_p^*>0$. Recall that $\mathrm{Res}$ is diagonal and for any loser $i\notin D$,
		\[
		\frac{\partial \dot z_i}{\partial z_i}\Big|_{z=z_{e_4}^*}
		=1+w_i-k\sum_{p,q\in D}z_p^*z_q^*
		=1+w_i-k r^2.
		\]
		Thus, the equilibrium $z_{e_4}^*$ is locally asymptotically stable if and only if
		\[
		1+w_i-k r^2<0,\qquad \forall\, i\notin D.
		\]
		In particular, when $k<1$, there may exist locally asymptotically stable equilibria with $d\ge 2$ (WSA).
		On the other hand, for $d=1$ (a single winner $j$), the above condition becomes
		$1+w_i-k(1+w_j)<0$ for all $i\neq j$, which shows that a WTA equilibrium may also be locally asymptotically stable
		if the input of the winner is sufficiently larger than the others.
		Therefore, when $k<1$, the competitive outcome can be either WSA or WTA.
	\end{proof}
	
	\section{\textsc{The global stability of the WTA solution when $k=1$}.}\label{global-stability-WTA}
	In this section, we aim to discuss the global stability of System (\ref{model-1}) when $k=1$. In this case, the original system can be rewritten in the following form:
	\begin{equation}\label{k=1-model}
		\dot{z}_\iota=z_i(1+w_i-\sum_{p,q=1}^{n}z_pz_q),
	\end{equation}
	
	Without loss of generality, we may assume that the external inputs $w_i$ follow the inequality below with respect to their magnitudes: $w_1> w_2\geq w_3\geq\cdots\geq w_{n-1}\geq w_n\geq0$.
	
	\begin{theorem}\label{global-WTA}
		System (\ref{k=1-model}) will evolve into a WTA competition outcome under any initial strictly positive conditions, and the globally stable equilibrium point will be $\left(\sqrt{1+w_{max}},0,\cdots,0\right)^T$.
	\end{theorem}
	
	\begin{proof}
		Consider strictly positive initial conditions $z(0)$.
		For system (\ref{k=1-model}) with $k=1$, we can write
		\[
		\dot z_i=z_i\Big(1+w_i-\sum_{p,q=1}^{n}z_pz_q\Big),\qquad i=1,\dots,n.
		\]
		Denote
		\[
		S(t)=\sum_{p,q=1}^{n}z_p(t)z_q(t)=\Big(\sum_{\ell=1}^{n}z_\ell(t)\Big)^2.
		\]
		
		For any $i\neq 1$, since $z_i(t)>0$ and $z_1(t)>0$ for all $t$,
		$\frac{d}{dt}\ln\frac{z_i(t)}{z_1(t)}
		=\frac{\dot z_i}{z_i}-\frac{\dot z_1}{z_1}
		=\big(1+w_i-S(t)\big)-\big(1+w_1-S(t)\big)=w_i-w_1<0.$
		
		Thus,
		\begin{equation}\label{eq:ratio_exp}
			\frac{z_i(t)}{z_1(t)}=\frac{z_i(0)}{z_1(0)}\,e^{(w_i-w_1)t},
			\qquad i=2,\dots,n,
		\end{equation}
		which implies
		\begin{equation}\label{eq:ratio_to_zero}
			\lim_{t\to\infty}\frac{z_i(t)}{z_1(t)}=0,\qquad i=2,\dots,n.
		\end{equation}
		
		Note that $S(t)\ge z_1(t)^2$ and hence
		\[
		\dot z_1(t)=z_1(t)\big(1+w_1-S(t)\big)\le z_1(t)\big(1+w_1-z_1(t)^2\big).
		\]
		It follows that $z_1(t)$ is bounded above by the solution $y(t)$ of
		\[
		\dot y(t)=y(t)\big(1+w_1-y(t)^2\big),\qquad y(0)=z_1(0).
		\]
		In particular,
		\[
		z_1(t)\le y(t)\le \max\Big\{z_1(0),\,\sqrt{1+w_1}\Big\},\qquad \forall\, t\ge 0.
		\]
		
		We next show that $\lim_{t\to\infty}z_1(t)$ indeed exists.
		From \eqref{eq:ratio_to_zero}, we have $\sum_{i=2}^n z_i(t)=o(z_1(t))$ as $t\to\infty$.
		Hence,
		$S(t)=\Big(z_1(t)+\sum_{i=2}^n z_i(t)\Big)^2
		= z_1(t)^2\Big(1+\frac{\sum_{i=2}^n z_i(t)}{z_1(t)}\Big)^2
		= z_1(t)^2\big(1+o(1)\big).$
		Fix an arbitrary $\delta>0$. Since $o(1)\to 0$, we can choose $\varepsilon\in(0,1)$ such that
		\[
		\sqrt{\frac{1+w_1}{1+\varepsilon}} \ge \sqrt{1+w_1}-\delta.
		\]
		Then there exists $T_\delta>0$ such that for all $t\ge T_\delta$,
		\[
		S(t)\le (1+\varepsilon)\,z_1(t)^2,
		\]
		and thus
		\[
		\dot z_1(t)=z_1(t)\big(1+w_1-S(t)\big)\ge z_1(t)\big(1+w_1-(1+\varepsilon)z_1(t)^2\big).
		\]
		In particular, whenever $0<z_1(t)\le \sqrt{1+w_1}-\delta$ and $t\ge T_\delta$, we have $\dot z_1(t)>0$.
		
		On the other hand, whenever $z_1(t)\ge \sqrt{1+w_1}+\delta$, we have
		\[
		S(t)=\Big(\sum_{i=1}^n z_i(t)\Big)^2\ge z_1(t)^2>(\sqrt{1+w_1}+\delta)^2>1+w_1,
		\]
		so $1+w_1-S(t)<0$ and hence $\dot z_1(t)<0$.

		Consequently, for all $t\ge T_\delta$, the vector field points strictly inside the interval
		$\big(\sqrt{1+w_1}-\delta,\sqrt{1+w_1}+\delta\big)$ on its boundary:
		$\dot z_1(t)>0 \ \text{whenever}\ z_1(t)\le \sqrt{1+w_1}-\delta,
		\qquad
		\dot z_1(t)<0 \ \text{whenever}\ z_1(t)\ge \sqrt{1+w_1}+\delta.$
		Since $z_1(t)$ is continuous and bounded, it cannot remain outside this interval for arbitrarily large times;
		otherwise it would be strictly increasing (resp. decreasing) on an unbounded time set while staying bounded, a contradiction.
		Hence, there exists $\tilde T_\delta\ge T_\delta$ such that
		\[
		\sqrt{1+w_1}-\delta \le z_1(t)\le \sqrt{1+w_1}+\delta,\qquad \forall\, t\ge \tilde T_\delta.
		\]
		Since $\delta>0$ is arbitrary, it follows that $\lim_{t\to\infty}z_1(t)$ exists.
		
		Let $L=\lim_{t\to\infty}z_1(t)$. By \eqref{eq:ratio_to_zero}, we have $z_i(t)\to 0$ for all $i\ge 2$.
		Hence $S(t)=\big(\sum_{\ell=1}^n z_\ell(t)\big)^2\to L^2$.
		Taking limits in $\dot z_1=z_1(1+w_1-S(t))$ yields
		\[
		0=\lim_{t\to\infty}\dot z_1(t)=L\big(1+w_1-L^2\big),
		\]
		so $L=\sqrt{1+w_1}$ (since $z_1(t)>0$).
		Therefore, $z_1(t)\to \sqrt{1+w_1}$ and \eqref{eq:ratio_to_zero} implies $z_i(t)\to 0$ for all $i=2,\dots,n$.
		
		Thus, system (\ref{k=1-model}) converges to $\left(\sqrt{1+w_{max}},0,\cdots,0\right)^T$.
	\end{proof}

	\section{Extension to the competition over uniform hypergraphs of arbitrary order}
	Now, we consider an $n$-neuron system with uniform arbitrary-body high-order interactions, the firing rate dynamics can be formulated as:
	\begin{equation}\label{model-3} 
		\begin{aligned}
			\frac{dz_{i}}{dt} ={} & z_{i}\big(1 - z_{i}^{t-1} 
			- k \!\!\sum_{(p_2,\cdots,p_t)\neq(i,\cdots,i)} z_{p_2} \cdots z_{p_t} \\
			& \quad + w_{i} \big), \quad i = 1,2,\cdots,n.
		\end{aligned}
	\end{equation}
	Neurons $p_2,\cdots, p_t$ are neighbors connected to $i$ through hyperedges. Similarly, the dynamics can be rewritten as:
	
	\begin{equation}\label{model-3-tensor} 
		\dot{z}=diag(z)(b+Cz^{t-1}),
	\end{equation}
	
	Firstly, we can derive the corresponding results for the number of the winners.
	
	\begin{theorem}\label{theorem-2u-equilibrium-points}
		For system (\ref{model-3}), when $k\geq1$, there can only be one winner, i.e. $d=1$. When $k<1$, multiple winners may exist, i.e. $d\geq 1$. 
	\end{theorem}
	
	\begin{proof}
		The proof is the same to the proof of Theorem \ref{theorem-1-equilibrium-points}.
	\end{proof}
	
	\begin{theorem}\label{theorem-7-equilibrium-points}
		System \eqref{model-3} has one unique positive equilibrium $z^*$ if $-C$ is an irreducible nonnegative $\mathcal{H}^{+}$-tensor. The equilibrium point $z^*$ is globally asymptotically stable.
	\end{theorem}
	
\begin{proof}
	Consider system \eqref{model-3} with the equilibrium condition
	\[
	b + C(z^*)^{t-1}=0,
	\quad\text{equivalently}\quad
	(-C)(z^*)^{t-1}=b.
	\]
	Under the assumption that $-C$ is an irreducible nonnegative $H^+$-tensor, the tensor equation
	$(-C)x^{t-1}=b$ admits a unique positive solution, hence there exists a unique all-neurons-coexistence equilibrium point
	$z^*$.
	
	We next show its global asymptotic stability.
	The proof follows the same Lyapunov construction as in Theorem~\ref{theorem-3-all-neurons-coexistence-equilibrium-points}.
	Specifically, in the Lyapunov function used in Theorem~\ref{theorem-3-all-neurons-coexistence-equilibrium-points}, we replace every occurrence of the
	quadratic term (power $2$) by the $(t-1)$-th power.
	With this modification, the Lyapunov function remains positive definite with respect to $z^*$ and radially unbounded
	on $\mathbb R_{+}^{n}$.
	
	Along the trajectories of system \eqref{model-3}, the time derivative of the modified Lyapunov function
	has exactly the same sign structure as in Theorem~\ref{theorem-3-all-neurons-coexistence-equilibrium-points}.
	The only difference is that differentiating the $(t-1)$-th power introduces a constant factor $(t-1)$,
	which is positive and therefore does not affect the sign.
\end{proof}

\begin{theorem}\label{theorem-2u-all-neurons-coexistence-equilibrium-points}
	Assuming that system (\ref{model-3}) with at least two species has a all-neurons-coexistence equilibrium point  $z_{e_3}^*=(z_{1}^*, z_{2}^*, \cdots, z_{n}^*)^T$, where $z_{i}^*>0$ for all neuron $i$. The equilibrium point  $z_{e_3}^*$ is locally asymptotically stable if and only if $k<1$.
\end{theorem}

\begin{proof}
	The proof is similar to the proof of Theorem \ref{theorem-2-all-neurons-coexistence-equilibrium-points}. The only difference is that
	\[
	\begin{aligned}
		T_i = -z_i^* \Big(&
		k \sum_{q_4,\cdots, q_t} z_{q_4}^* \cdots z_{q_t}^* 
		+ k \sum_{q_3,q_5,\cdots, q_t} z_{q_3}^* z_{q_5}^* \cdots z_{q_t}^* \\
		&+ \cdots 
		+ k \sum_{q_3,\cdots, q_{t-1}} z_{q_3}^* \cdots z_{q_{t-1}}^*
		\Big)
	\end{aligned}
	\]
	and \[\hat{Z}=(t-1)(k-1)\left(\begin{matrix}z_{1}^{*^{t-1}}& 0 & \cdots & 0\\ 0 &z_{2}^{*^{t-1}}& \cdots & 0\\\vdots&\vdots &\ddots& \vdots\\ 0& 0& \cdots &z_{n}^{*^{t-1}}\end{matrix}\right).\]
	The remaining proof remains similar as in Theorem \ref{theorem-2-all-neurons-coexistence-equilibrium-points}.
\end{proof}

Notice that the Jacobian structure is similar to the previous case and the similar proof strategies directly apply. We then further have the following results.

\begin{theorem}\label{three-result-u}
	When $k=1$, system (\ref{model-3}) evolves through competition into a WTA phenomenon, with the sole winner having the highest external input $W_{max}$. When $k>1$, system (\ref{model-3}) exhibits VWTA behavior after competition. When $k<1$, the system (\ref{model-3})'s competitive outcome is either WSA or WTA.
\end{theorem}

Without loss of generality, we may assume that the external inputs $w_i$ follow the inequality below with respect to their magnitudes: $w_1> w_2\geq w_3\geq\cdots\geq w_{n-1}\geq w_n\geq0$.

\begin{theorem}\label{global-WTA-u}
	Consider system (\ref{model-3}) with interaction order $t\ge 3$ and $k=1$.
	Assume that the external inputs satisfy
	$w_1> w_2\geq w_3\geq\cdots\geq w_{n-1}\geq w_n\geq0$.
	Then, under any initial strictly positive conditions, system (\ref{model-3}) will evolve into a WTA competition,
	and the globally stable equilibrium point is
	$\left((1+w_{max})^{\frac{1}{t-1}},0,\cdots,0\right)^T.$
\end{theorem}

In light of Proposition~2, we have the following result on the guaranteed existence bound for the case $k<1$.

\begin{corollary}\label{cor:exist_d_bound_klt1}
Consider system (\ref{model-3}).	Let $k\in(0,1)$ and $p=t-1\ge 2$. Define $b_i=1+w_i>0$ and
	\[
	b_{\min}=\min_i b_i,\qquad b_{\max}=\max_i b_i.
	\]
	If
	\begin{equation}\label{eq:d_exist_bound}
		d \;<\;\Bigg(\frac{(1-k)\,b_{\min}}{k\,(b_{\max}-b_{\min})}\Bigg)^{1/p},
	\end{equation}
	then the system admits at least one equilibrium with exactly $d$ winners (i.e., $d$ positive components and $n-d$ zeros).
\end{corollary}

\begin{proof}
	Fix any winner set $D$ with $|D|=d$ and write $x=z_D\in\mathbb R_{++}^d$.
	The equilibrium equations on $D$ reduce to the all-ones tensor equation
	\[
	k\,\mathcal E\,x^{p} + (1-k)\,I x^{p} = b_D,
	\]
	which is Proposition~2 with $a=k$, $s=1-k>0$, and dimension $n=d$.
	Applying the sufficient existence condition in Proposition~2,
	\[
	\frac{a\,d^{p}}{s}\,(\max_{i\in D}b_i-\min_{i\in D}b_i) < \min_{i\in D}b_i,
	\]
	and using $\min_{i\in D}b_i\ge b_{\min}$ and $\max_{i\in D}b_i-\min_{i\in D}b_i\le b_{\max}-b_{\min}$,
	we see that \eqref{eq:d_exist_bound} implies the above inequality and hence yields a positive solution $x\gg 0$.
	Setting $z_D=x$ and $z_j=0$ for $j\notin D$ gives an equilibrium with exactly $d$ winners.
\end{proof}

\begin{remark}\label{rem:wk_large_d_short}
	The large-$d$ equilibria are favored when (i) $w$ is nearly homogeneous (small $w_{\max}-w_{\min}$) and/or (ii) $k$ is
	small (weaker higher-order coupling). Conversely, larger heterogeneity or $k\to 1$ shrinks the guaranteed range of $d$. In summary, as $k$ increases from $0$ to $1$, the system transitions from favoring full coexistence to favoring
	partial coexistence, and eventually to a WTA-type regime.  Moreover, the role of the interaction order $p=t-1$ is that this shrinkage is \emph{attenuated} when $p$ is large and thus makes this transition more gradual.
\end{remark}

\begin{remark}\label{remark-t2-lv}
	The interaction order $t$ provides a unified viewpoint that includes the classical pairwise model as a special case.
	In particular, when $t=2$, the higher-order term reduces to a pairwise interaction term and system (\ref{model-3})
	degenerates to the standard competitive Lotka--Volterra (LV) dynamics
	\[
	\dot z_i
	= z_i\Big(1+w_i - z_i - k\!\!\sum_{j\neq i} z_j\Big)
	= z_i\Big((1+w_i)-\sum_{j=1}^n a_{ij}z_j\Big),
	\]
	where $a_{ii}=1$ and $a_{ij}=k$ for $i\neq j$.
	Thus, the proposed framework recovers the classical LV competition model when the interaction order is reduced to $t=2$.
\end{remark}

\begin{remark}\label{t}
	These theorems highlight that the type of competition outcome, whether WTA, VWTA, or WSA, is fundamentally determined by the ratio between self-inhibition and lateral-inhibition, denoted by $k$. However, the interaction-order $t$ influences only the equilibrium values of neuronal states or the maximum winner number without altering the qualitative competition outcome. Therefore, the system exhibits a structurally invariant competition mechanism with respect to the interaction-order. Higher-order interactions modulate the system dynamics but do not change the competition outcome determined by $k$.
\end{remark}




\section{Application: A Dynamic WTA Protocol}\label{Application}

In this section, we explore the design of distributed protocols for solving WTA, WSA, and VWTA problems over hypergraphs, by utilizing the states of locally interacting agents. The proposed framework is inspired by biologically realistic competition dynamics and offers a flexible, extensible, and fully decentralized approach.

We consider a system composed of multiple nodes, and each node is modeled as an autonomous agent. The evolution of each agent is governed by the following general dynamics:
\[
\dot{z}_i = f_i\left(z_i, z_{j \in \mathcal{N}(i)}, u_i\right),
\]
where $\mathcal{N}(i)$ denotes the set of neighbors of node $i$ defined by the hypergraph topology. Here, $z_i$ is the state of agent $i$, $z_{j \in \mathcal{N}(i)}$ are the states of neighboring nodes, and $u_i$ denotes an external input or bias term acting on node $i$. Notably, the dynamics of node $i$ depend solely on locally available information, such as its own state, the states of its neighbors, and its input. This feature makes the proposed protocol especially suitable for large-scale, resource-constrained, and privacy-sensitive environments, where minimizing communication spending, efficient use of limited computational resources, and preservation of data confidentiality are essential. 

Based on the theoretical framework developed in this paper, the function $f_i(\cdot)$ can be instantiated as the competition models described in \eqref{model-1} and \eqref{model-3}. These models characterize nonlinear, high-order interactions among neurons (or agents) and are well suited for deployment in distributed systems with complex interaction topologies.

By tuning the parameter $k$, the system can demonstrate a variety of competitive outcomes. This tunability enables the protocol to flexibly accommodate different selection modes, including WTA or VWTA behaviors where only one neuron remains active, and WSA behavior where multiple active agents coexist. Such flexibility makes the protocol suitable for various applications requiring either exclusive selection of a single dominant agent or cooperative activation among multiple agents.

Overall, the proposed distributed protocol unifies mathematical soundness and biological inspiration, opening new possibilities for practical decentralized decision-making, pattern selection, and consensus formation on hypergraph networks.

\section{Simulations}\label{Simulations}

In this section, we use numerical simulations to verify the validity of the theoretical results presented earlier.

Assume there are 10 neurons in system (\ref{model-1}), with initial input values \(w_0=(1.9557, 2.8322, 3.8317, 2.2795, 1.3796,\\ 7.1796, 8, 2.4356, 3.8469, 1.7953)^T\), and initial firing rates \(z_0=(0.1234, 0.1678, 0.1101, 0.1345, 0.1789, 0.6256, 0.1890,\\ 0.1567, 0.1910, 0.1346)^T\). The dynamics of these 10 neurons are described by equations (\ref{model-1}). It is noteworthy that, under these initial conditions, neuron 7 has the largest initial input $8$.

\begin{example}\label{example-3}
	We performed numerical simulations on the ten neurons in system (\ref{model-1}) using initial values $w_0$ and $z_0$. When	$k=0.5$, the state trajectories of the ten neurons are shown in Fig.\ref{ex-3}. The competition result evolved by the system (\ref{model-1}) is WSA, which is consistent with the conclusion of Theorems \ref{theorem-1-equilibrium-points} and \ref{three-result}. When $k=0.01$, the state trajectories of the ten neurons are presented in Fig.\ref{ex-4}. It is easy to see that the system ultimately converges to the all-neurons-coexistence equilibrium point, which is fully consistent with Theorems \ref{theorem-2-all-neurons-coexistence-equilibrium-points} and \ref{theorem-3-all-neurons-coexistence-equilibrium-points}.
\end{example}
\begin{remark}\label{simulation-WSA}
	Through simulations, we observed that when $0<k<1$, system (\ref{model-1}) evolves to have multiple winners, i.e., WSA, and the number of winners increases as the value of $k$ decreases. Thus, if the self-inhibition of neurons is much greater than the lateral-inhibition, it is more favorable for the neurons to coexist in the competition.
\end{remark}

\begin{figure}[!htb]
	\centering
	\includegraphics[width=2.3in]{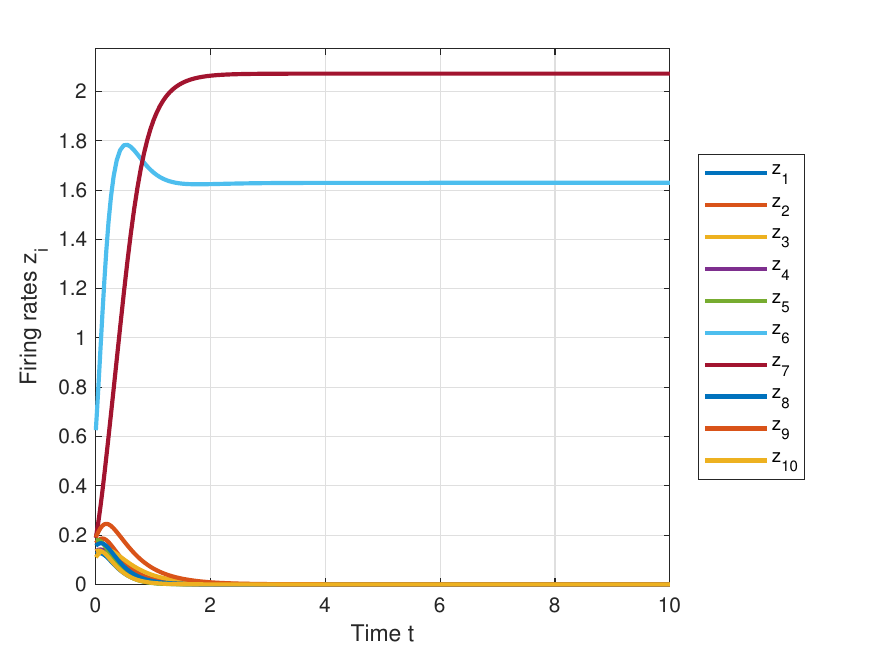}\\
	\caption{The state trajectories of all neurons in system (\ref{model-1}) when $k=0.5$.}\label{ex-3}
\end{figure}

\begin{figure}[!htb]
	\centering
	\includegraphics[width=2.3in]{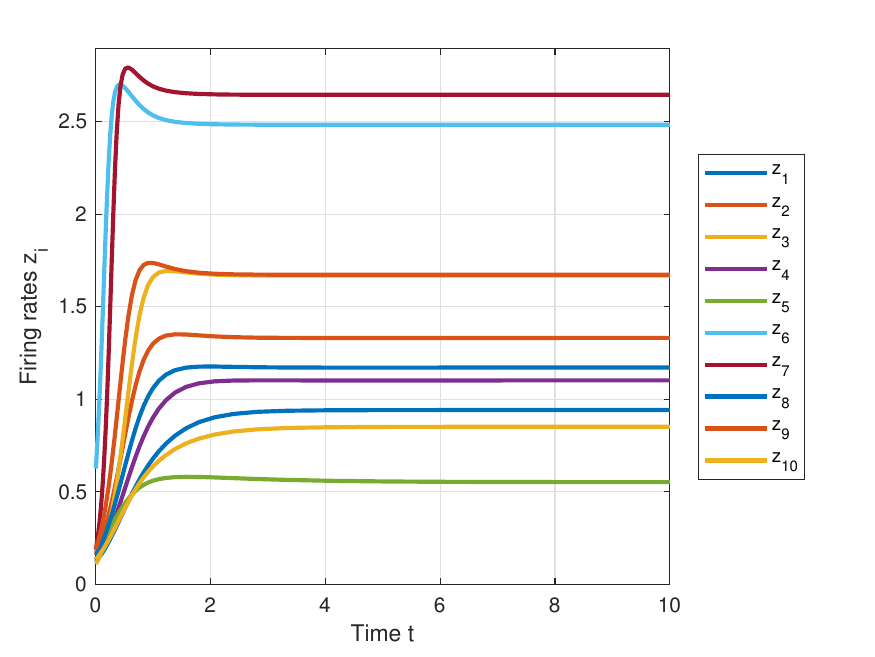}\\
	\caption{The state trajectories of all neurons in system (\ref{model-1}) when $k=0.01$.}\label{ex-4}
\end{figure}

\begin{example}\label{example-2}
	When $k=1.5$, we consider the numerical simulation of system (\ref{model-1}) with initial values $w_0$ and $z_0$. Fig.\ref{ex-2} shows the evolution of the states values along with time for all neurons. It is clearly observed that when $k>1$, there is only one winner, i.e. neuron 6, while all other neurons become inactive. Notably, neuron 6 is not the neuron with the largest initial input $W_{max}$. This simulation result partially validates Theorems \ref{theorem-1-equilibrium-points} and \ref{three-result}, demonstrating that when $k>1$, system (\ref{model-1}) exhibits VWTA behavior.
\end{example}
\begin{figure}[!t]
	\centering
	\includegraphics[width=2.3in]{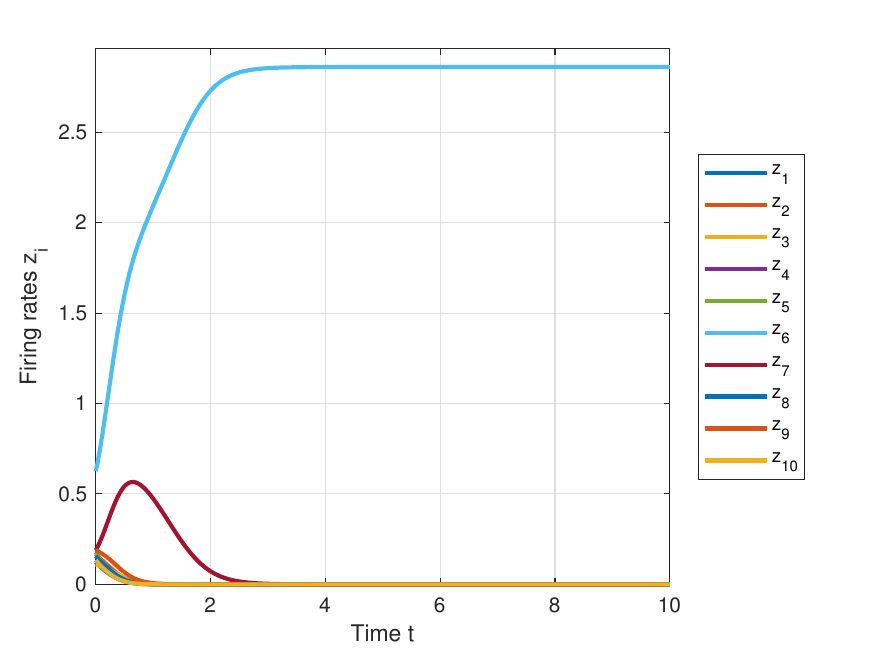}\\
	\caption{\centering The state trajectories of all neurons in system (\ref{model-1}) when $k=1.5$.}\label{ex-2}
\end{figure}

\begin{example}\label{example-1}
	When $k=1$, perform numerical simulations on system (\ref{model-1}) while keeping the initial values unchanged. The evolution of the state values of the 10 neurons over time is shown in Fig.\ref{ex-1}. It is evident that the neuron $7$ with the largest input $8$, emerges as the sole winner, thereby exemplifying the WTA phenomenon. Concurrently, neuron $7$ ultimately converges to a value of $\sqrt{1+w_{7}}=3$. The simulation results are fully consistent with Theorem \ref{global-WTA} and partially validate Theorems \ref{theorem-1-equilibrium-points} and \ref{three-result}.

\end{example}
\begin{figure}[!t]
	\centering
	\includegraphics[width=2.3in]{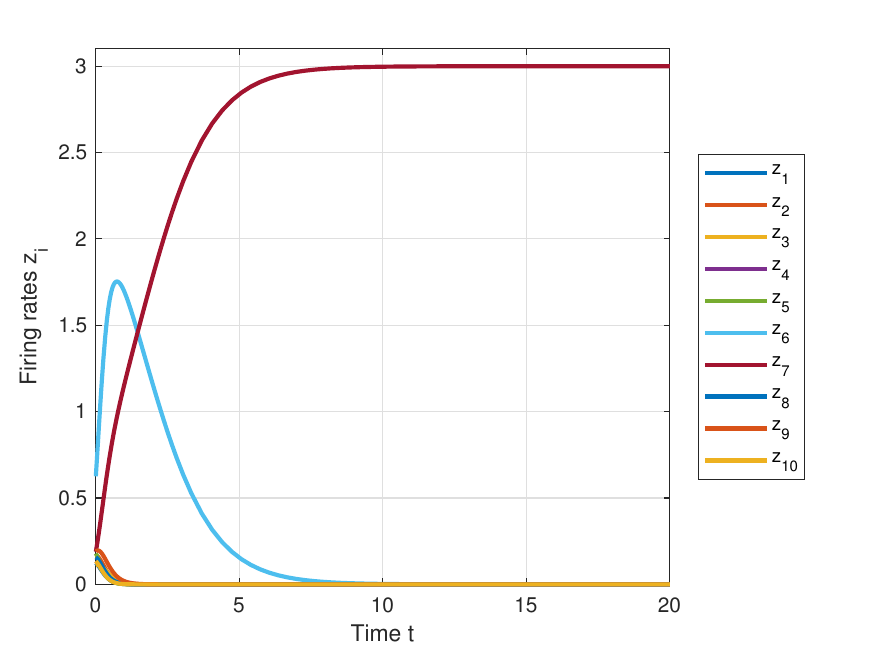}\\
	\caption{The state trajectories of all neurons in system (\ref{model-1}) when $k=1$.}\label{ex-1}
\end{figure}

\begin{example}\label{example-4}
	For system (\ref{model-3}), we set the interaction order $t=5$ and adopt the same initial conditions $w_0$ and $z_0$ as above. Numerical simulations were performed to analyze the state evolution of each neuron under different values of $k$, specifically $k=0.0001$ (Fig.\ref{ex-5-1}), $k=0.5$ (Fig.\ref{ex-5-2}), $k=1$ (Fig.\ref{ex-5-3}), and $k=2$ (Fig.\ref{ex-5-4}). The simulation results confirm the validity of Theorems \ref{theorem-2u-equilibrium-points}, \ref{theorem-7-equilibrium-points}, \ref{theorem-2u-all-neurons-coexistence-equilibrium-points}, and \ref{three-result-u}.
\end{example}
\begin{figure}[!t]
	\centering
	\includegraphics[width=2.3in]{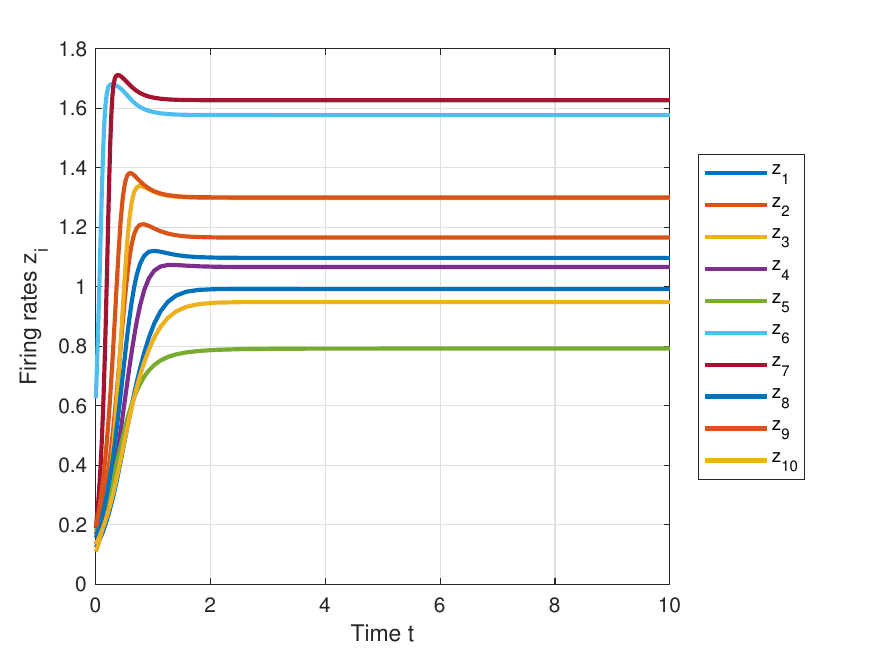}\\
	\caption{The state trajectories of all neurons in system (\ref{model-3}) when $k=0.0001$.}\label{ex-5-1}
\end{figure}
\begin{figure}[!t]
	\centering
	\includegraphics[width=2.3in]{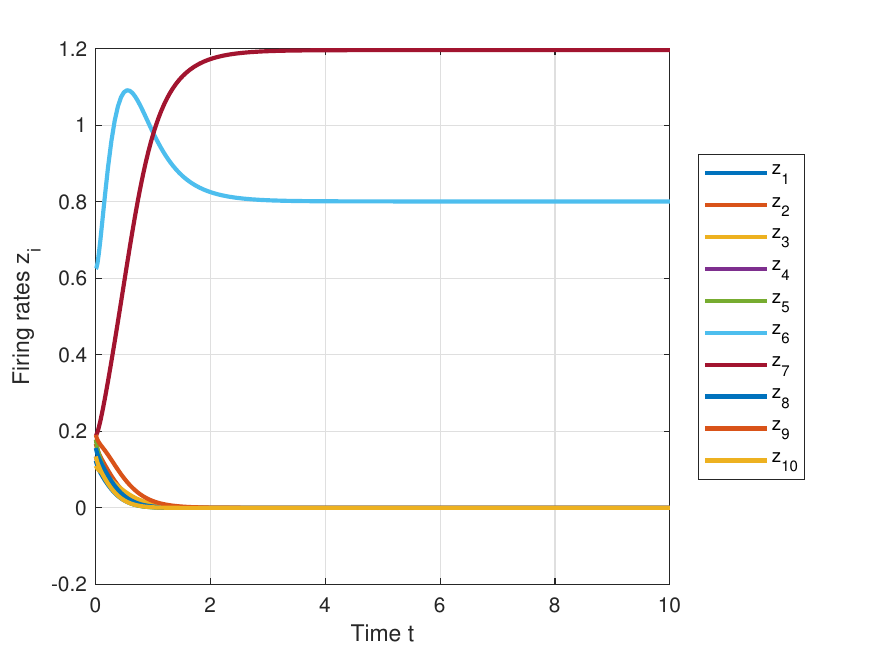}\\
	\caption{The state trajectories of all neurons in system (\ref{model-3}) when $k=0.5$.}\label{ex-5-2}
\end{figure}
\begin{figure}[!t]
	\centering
	\includegraphics[width=2.3in]{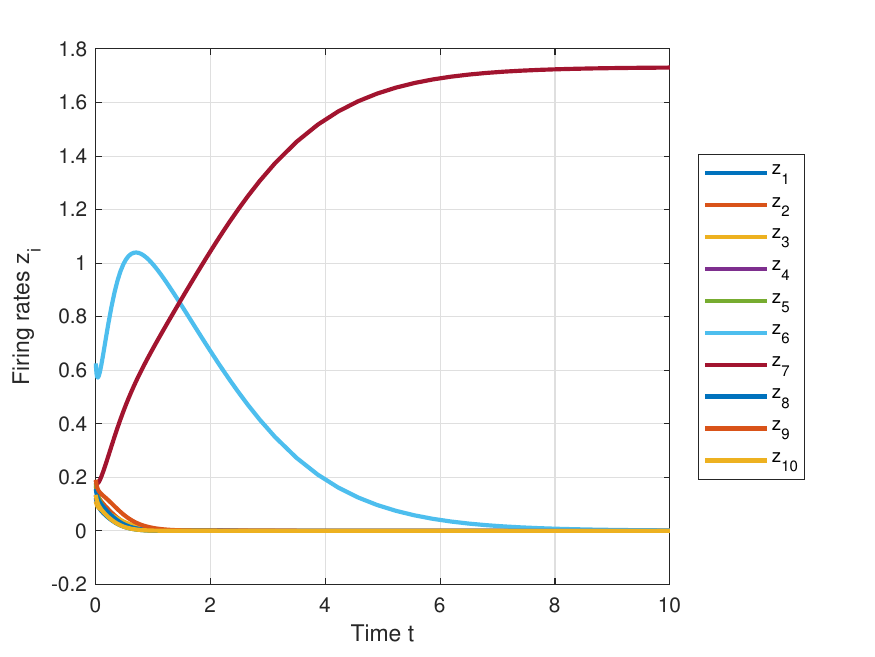}\\
	\caption{The state trajectories of all neurons in system (\ref{model-3}) when $k=1$.}\label{ex-5-3}
\end{figure}
\begin{figure}[!htb]
	\centering
	\includegraphics[width=2.3in]{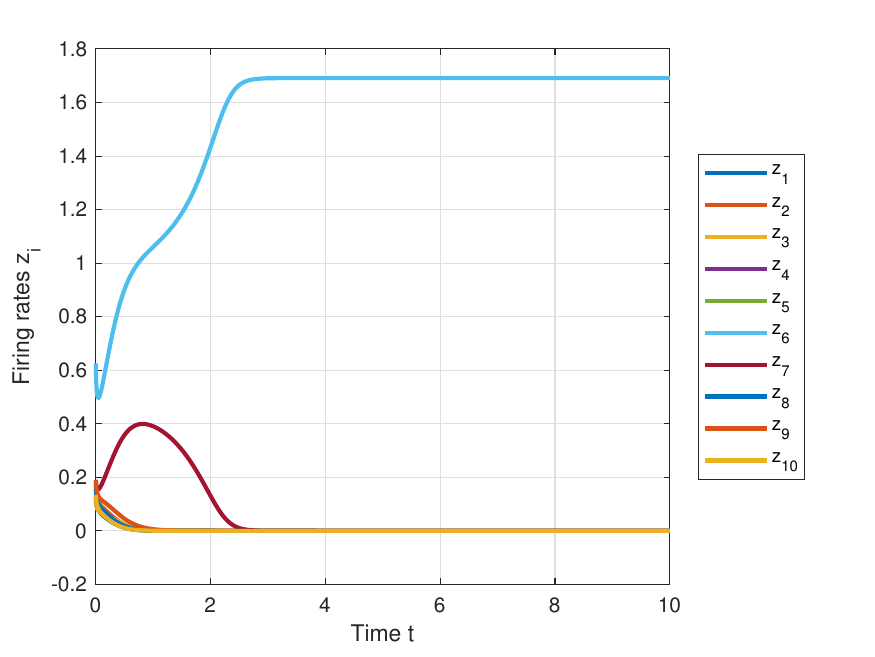}\\
	\caption{The state trajectories of all neurons in system (\ref{model-3}) when $k=2$.}\label{ex-5-4}
\end{figure}

\begin{remark}\label{simulation-t=5}
	For system (\ref{model-3}), we present the simulation results for the case $t=5$ as a representative example to illustrate the evolution of the neuronal state under higher-order network structures. While the simulations continue to support the theoretical results established in this paper, it is evident that achieving all-neurons-coexistence or the VWTA pattern requires a larger or smaller parameter $k$ compared to the case with $t=3$ (system (\ref{model-1})). This suggests that the lower the interaction-order of the network, the more sensitive the neuronal dynamics are to the ratio between self-inhibition and lateral-inhibition. Moreover, the stable states of individual neurons vary in all-neurons-coexistence, WSA and VWTA cases, indicating that while increasing the interaction-order does not alter the competition pattern, it does affect the final stable states of the neurons.
\end{remark}

\begin{example}\label{example-5}
	For traditional pairwise interaction system $\frac{dz_{i}}{dt}=z_{i}(1-z_{i}-k\sum_{j\neq i}z_{j}+w_{i}), i=1,2,\cdots,10$, we reuse the initial conditions $w_0$ and $z_0$ introduced above. The evolution of neuronal states was examined through numerical simulations at varying values of the parameter $k$. Specifically, the cases of $k = 0.0001$, $0.5$, $1$, and $2$ are presented in Fig.\ref{ex-6-1}--\ref{ex-6-4}, respectively. Under identical initial conditions and parameters, the proposed model with higher-order interactions exhibits the same competition patterns as those generated by traditional pairwise interaction networks: WTA for $k=1$, WSA for $k<1$, and VWTA for $k>1$. However, the final steady states of individual neurons show some differences.
\end{example}
\begin{figure}[!htb]
	\centering
	\includegraphics[width=2.3in]{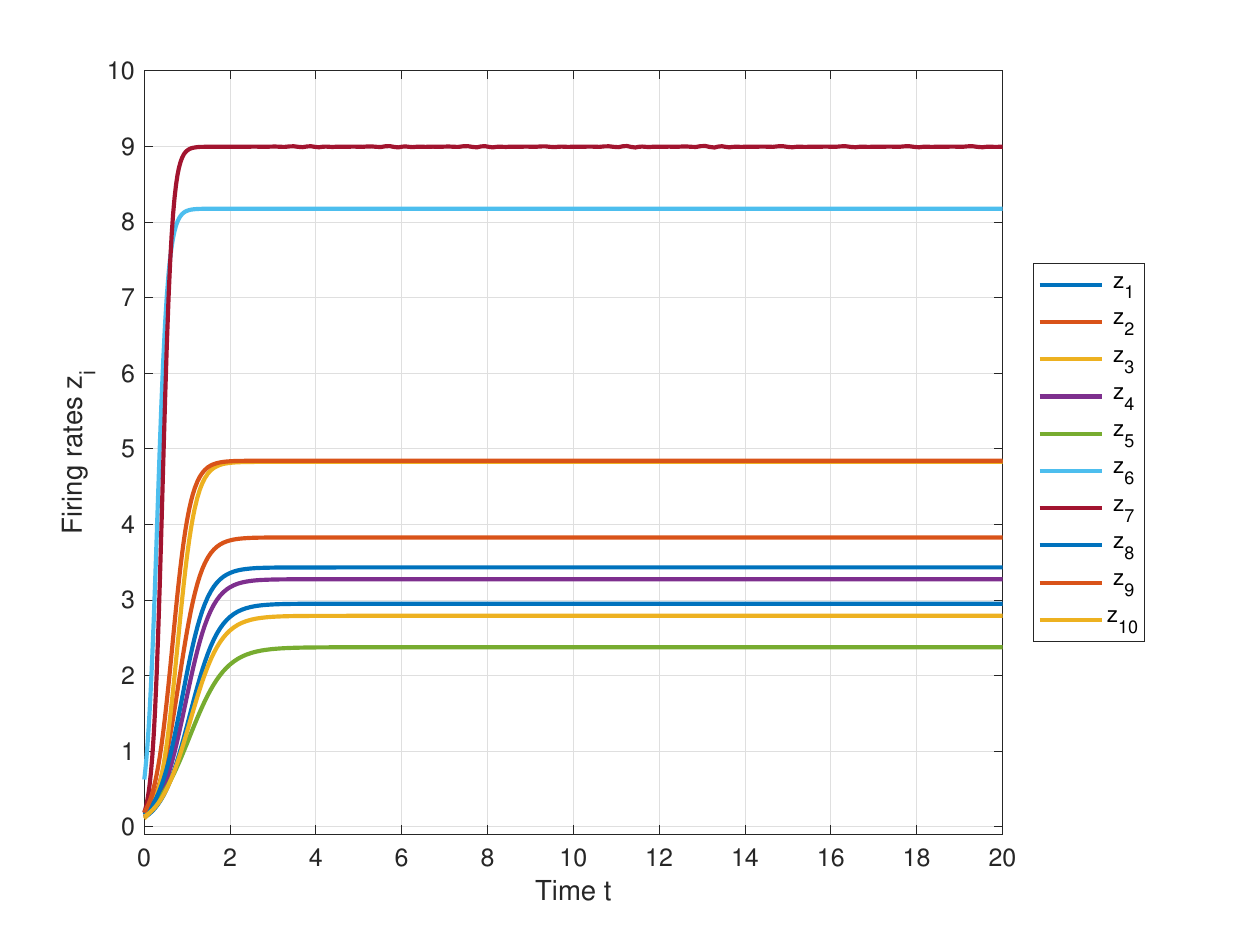}\\
	\caption{The state trajectories of all neurons in pairwise interaction system when $k=0.0001$.}\label{ex-6-1}
\end{figure}
\begin{figure}[!htb]
	\centering
	\includegraphics[width=2.3in]{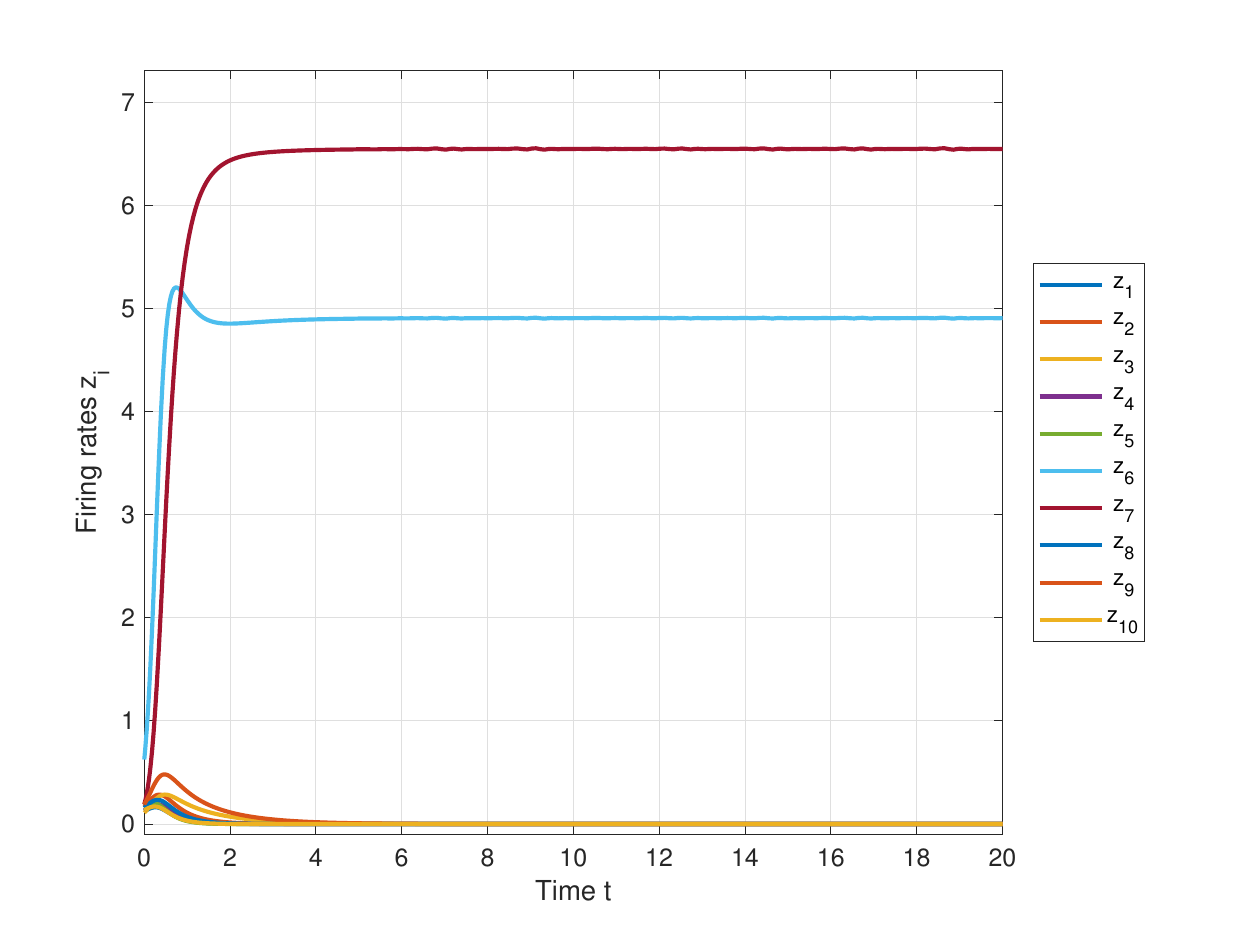}\\
	\caption{The state trajectories of all neurons in pairwise interaction system when $k=0.5$.}\label{ex-6-2}
\end{figure}
\begin{figure}[!htb]
	\centering
	\includegraphics[width=2.3in]{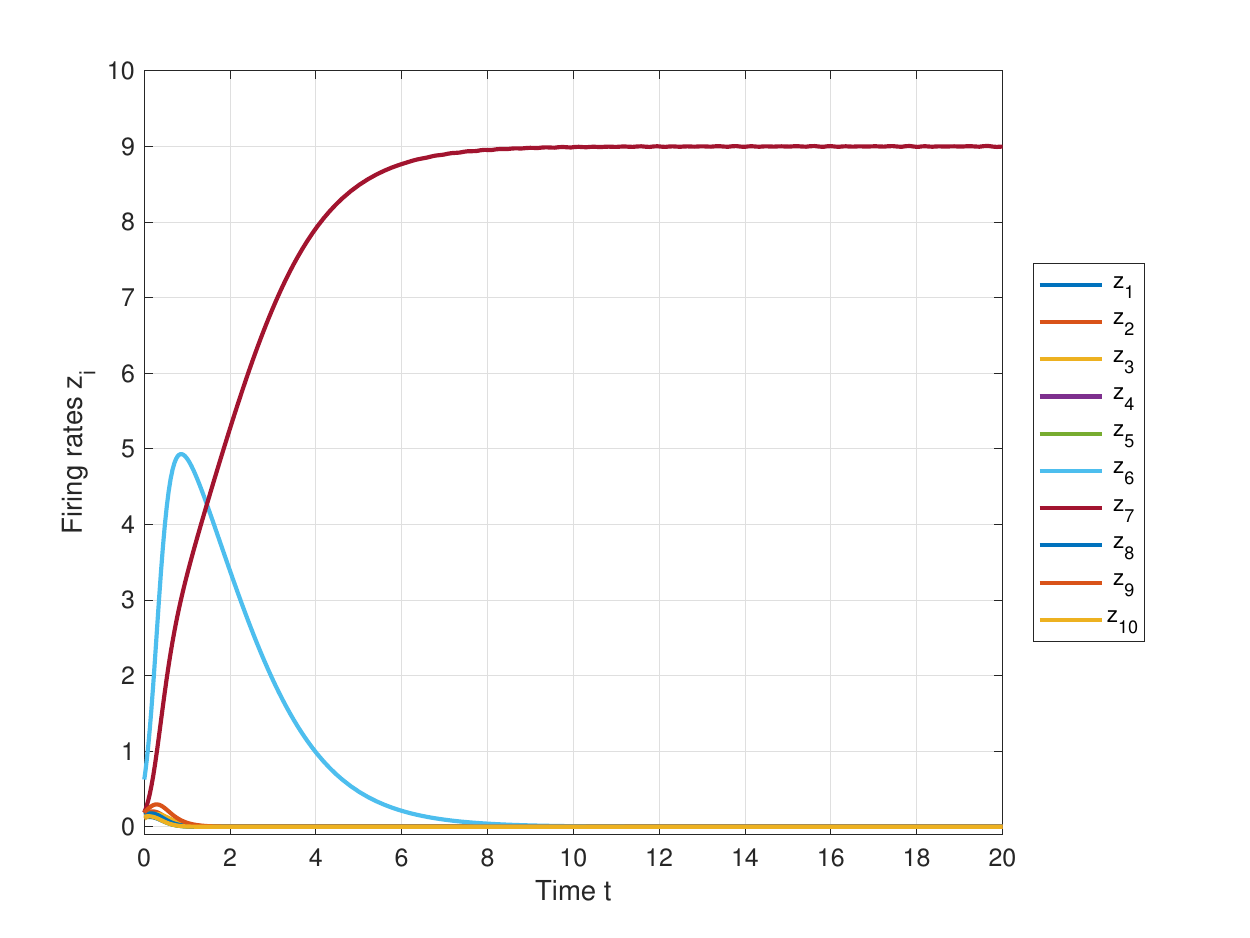}\\
	\caption{The state trajectories of all neurons in pairwise interaction system when $k=1$.}\label{ex-6-3}
\end{figure}
\begin{figure}[!htb]
	\centering
	\includegraphics[width=2.3in]{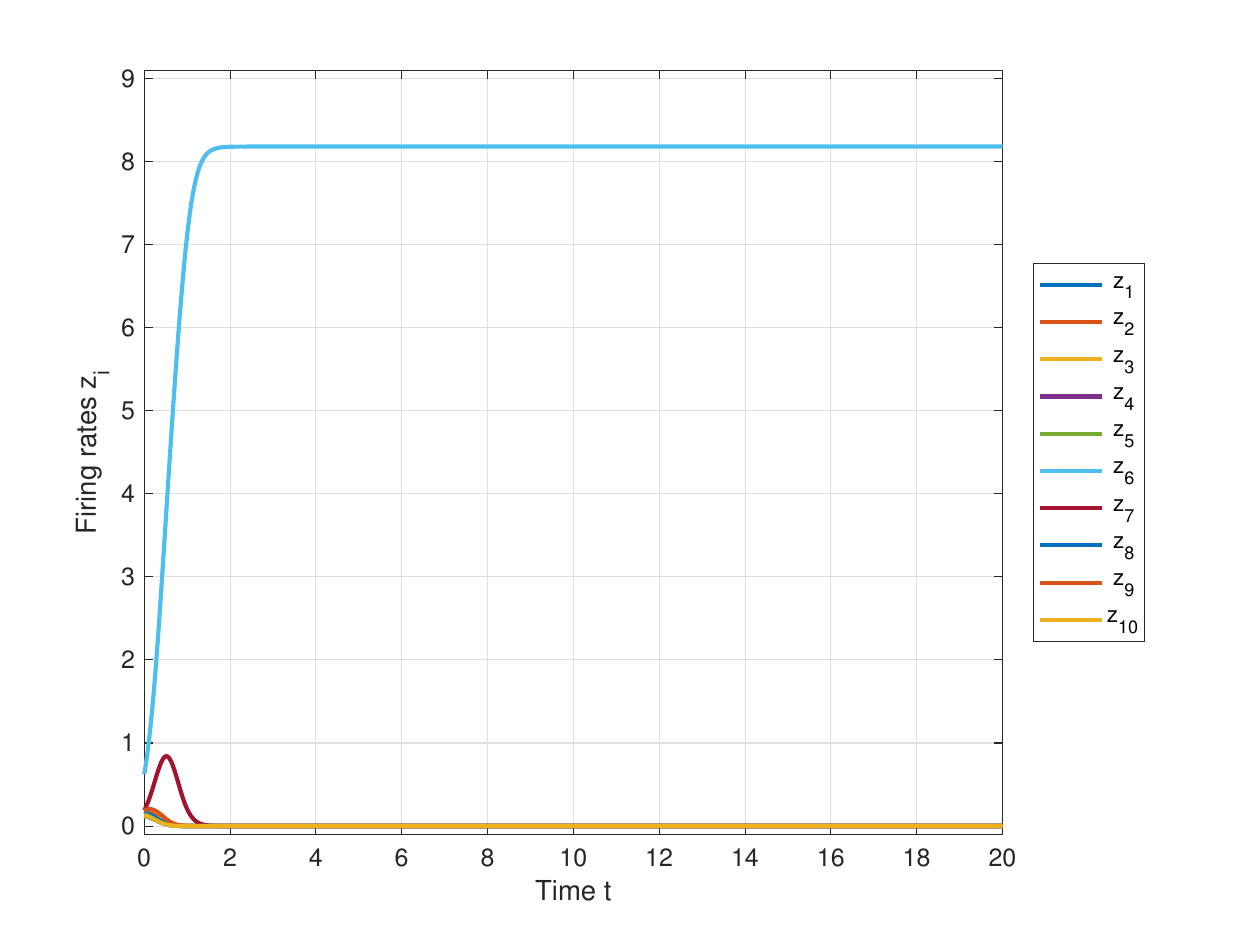}\\
	\caption{The state trajectories of all neurons in pairwise interaction system when $k=2$.}\label{ex-6-4}
\end{figure}

\begin{remark}\label{simulation-pairwise}
	Example \ref{example-5} demonstrates that introducing higher-order interactions does not alter the fundamental competition patterns (WTA, WSA, VWTA) under fixed initial conditions and parameters. This highlights the structural robustness of competition patterns to interaction-order, while revealing that such generalizations enable richer and more faithful representations of multi-neuron coordination beyond pairwise interaction models.
\end{remark}

\section{CONCLUSION}\label{conclusion}
In this paper, we studied WTA-type selection in a Lotka–Volterra competitive dynamics on higher-order networks, where classical pairwise inhibition is augmented by hyperedge-induced multi-way interactions described in a hypergraph/tensor form. We established rigorous results on the existence, uniqueness, and stability of equilibria corresponding to the WTA/WSA/VWTA outcomes using various mathematical tools, including LaSalle’s invariance principle, Lyapunov’s indirect method, properties of structured tensors and related tensor algebra methods. A key finding is an outcome-invariance property: the qualitative steady-state selection pattern is insensitive to the hyperedge order, and is instead governed by a small set of interpretable scalar parameters, notably the ratio between self-inhibition and lateral inhibition and the external inputs. These results confirm the designed model’s capacity to capture competitive neural activity in systems with high-order interaction structures. Future work will focus on extending the model to incorporate non-uniform hypergraphs that better reflect heterogeneous multi-neuron connectivity, exploring adaptive mechanisms that allow dynamic modulation of competitive interactions, and investigating the impact of stochastic fluctuations to capture intrinsic neural variability and noise effects in biologically realistic settings.


\begin{thebibliography}{1}
\bibliographystyle{abbrv}
	\bibitem{Han-2007}
	J. H. Han, S. A. Kushner, A. P. Yiu, et al., ``Neuronal competition and selection during memory formation,'' \textit{Science}, vol. 316, no. 5823, pp. 457--460, 2007.
	
	\bibitem{Tang-2024}
	S. Tang, L. Cui, J. Pan, and N. L. Xu, ``Dynamic ensemble balance in direct-and indirect-pathway striatal projection neurons underlying decision-related action selection,'' \textit{Cell Reports}, vol. 43, no. 9, p. 114726, 2024.
	
	\bibitem{Clark-2004}
	L. Clark, R. Cools, and T. Robbins, ``The neuropsychology of ventral prefrontal cortex: decision-making and reversal learning,'' \textit{Brain and Cognition}, vol. 55, no. 1, pp. 41--53, 2004.
	
	\bibitem{Lotka-1920}
	A. J. Lotka, ``Analytical note on certain rhythmic relations in organic systems,'' \textit{Proc. Natl. Acad. Sci. U.S.A.}, vol. 6, no. 7, pp. 410--415, 1920.
	
	\bibitem{Volterra-1928}
	V. Volterra, ``Variations and fluctuations of the number of individuals in animal species living together,'' \textit{ICES J. Mar. Sci.}, vol. 3, no. 1, pp. 3--51, 1928.
	
	\bibitem{Xiao-2017}
	Y. Xiao, M. T. Angulo, J. Friedman, et al., ``Mapping the ecological networks of microbial communities,'' \textit{Nat. Commun.}, vol. 8, no. 1, p. 2042, 2017.
	
	\bibitem{Pavel-2011}
	L. Pavel and L. Chang, ``Lyapunov-based boundary control for a class of hyperbolic Lotka--Volterra systems,'' \textit{IEEE Trans. Autom. Control}, vol. 57, no. 3, pp. 701--714, 2011.
	
	\bibitem{Asai-1999}
	T. Asai, M. Ohtani, and H. Yonezu, ``Analog integrated circuits for the Lotka-Volterra competitive neural networks,'' \textit{IEEE Trans. Neural Netw.}, vol. 10, no. 5, pp. 1222--1231, 1999.
	
	\bibitem{ZhangYi-2002}
	Z. Yi and K. K. Tan, ``Dynamic stability conditions for Lotka-Volterra recurrent neural networks with delays,'' \textit{Phys. Rev. E}, vol. 66, no. 1, Art. no. 011910, 2002.
	
	\bibitem{Es-saiydy-2022}
	M. Es-saiydy and M. Zitane, ``Oscillating dynamics of Lotka--Volterra neural networks with time-varying delays and distributed delays,'' \textit{Ricerche di Matematica}, vol. 73, no. 5, pp. 2779--2799, 2024.
	
	\bibitem{Bick-2017}
	C. Bick, ``Lotka--Volterra like dynamics in phase oscillator networks,'' in \textit{Advances in Dynamics, Patterns, Cognition: Challenges in Complexity}, vol. 20, pp. 115--125, 2017.
	
	\bibitem{Fukai-97}
	T. Fukai and S. Tanaka, ``A simple neural network exhibiting selective activation of neuronal ensembles: From winner-take-all to winners-share-all,'' \textit{Neural Comput.}, vol. 9, no. 1, pp. 77--97, 1997.
	
	\bibitem{Yi-2013}
	B. Zheng and Z. Yi, ``Using competitive layer model implemented by Lotka--Volterra recurrent neural networks for detecting brain activated regions from fMRI data,'' \textit{Neural Comput. Appl.}, vol. 22, no. 1, pp. 395--404, 2013.
	
	\bibitem{Bochuan-Zheng}
	B. Zheng, ``A winner-take-all Lotka--Volterra recurrent neural network with only one winner in each row and each column,'' \textit{Neural Comput. Appl.}, vol. 24, no. 1, pp. 1749--1757, 2014.
	
	\bibitem{Devgupta-2024}
	R. Devgupta, R. A. Dandekar, R. Dandekar, and S. Panat, ``Scientific machine learning in ecological systems: A study on the predator-prey dynamics,'' arXiv preprint arXiv:2411.06858, 2024.
	
	\bibitem{Li-2021}
	G. A. Maciel and R. Martinez-Garcia, ``Enhanced species coexistence in Lotka-Volterra competition models due to nonlocal interactions,'' \textit{J. Theor. Biol.}, vol. 530, p. 110872, 2021.
	
	\bibitem{Abrams-1983}
	P. A. Abrams, ``Arguments in favor of higher order interactions,'' \textit{Am. Nat.}, vol. 121, no. 6, pp. 887--891, 1983.
	
	\bibitem{Cui-SIS}
	S. Cui, F. Liu, L. Liang, H. Jardón-Kojakhmetov, and M. Cao, ``An SIS diffusion process with direct and indirect spreading on a hypergraph,'' \textit{Automatica}, vol. 158, p. 111035, 2025.
	
	\bibitem{Mayfield2017}
	M. M. Mayfield and D. B. Stouffer, ``Higher-order interactions capture unexplained complexity in diverse communities,'' \textit{Nat. Ecol. Evol.}, vol. 1, no. 3, p. 0062\%, 2017.
	
	\bibitem{Carletti2020}
	T. Carletti, D. Fanelli, and S. Nicoletti, ``Dynamical systems on hypergraphs,'' \textit{J. Phys. Complexity}, vol. 1, no. 3, p. 035006, 2020.
	
	\bibitem{Battiston2021}
	F. Battiston, E. Amico, A. Barrat, et al., ``The physics of higher-order interactions in complex systems,'' \textit{Nat. Phys.}, vol. 17, no. 10, pp. 1093--1098, 2021.
	
	\bibitem{Cui-DT2024}
	S. Cui, G. Zhang, H. Jardón-Kojakhmetov, and M. Cao, ``On discrete-time polynomial dynamical systems on hypergraphs,'' \textit{IEEE Control Syst. Lett.}, vol. 8, pp. 1078--1083, 2024.
	
	\bibitem{Cui2025Metzler}
	S. Cui, G. Zhang, H. Jardón-Kojakhmetov, and M. Cao, ``On Metzler positive systems on hypergraphs,'' \textit{IEEE Trans. Control Netw. Syst.}, vol. 12, no. 2, pp. 345--356, 2025.
	
	\bibitem{Letten-2019}
	A. D. Letten and D. B. Stouffer, ``The mechanistic basis for higher-order interactions and non-additivity in competitive communities,'' \textit{Ecol. Lett.}, vol. 22, no. 3, pp. 423--436, 2019.
	
	\bibitem{Sidhom-2024}
	L. Sidhom and T. Galla, ``Higher-order interactions in random Lotka-Volterra communities,'' arXiv preprint arXiv:2409.10990, 2024.
	
	\bibitem{Sales-Pardo-2023}
	M. Sales-Pardo, A. Mariné-Tena, and R. Guimerà, ``Hyperedge prediction and the statistical mechanisms of higher-order and lower-order interactions in complex networks,'' \textit{Proc. Natl. Acad. Sci. U.S.A.}, vol. 120, no. 50, e2303887120, 2023.
	
	\bibitem{Singh-2021}
	P. Singh and G. Baruah, ``Higher order interactions and species coexistence,'' \textit{Theor. Ecol.}, vol. 14, no. 1, pp. 71--83, 2021.
	
	\bibitem{Gibbs-2022}
	T. Gibbs, S. A. Levin, and J. M. Levine, ``Coexistence in diverse communities with higher-order interactions,'' \textit{Proc. Natl. Acad. Sci. U.S.A.}, vol. 119, no. 43, e2205063119, 2022.
	
	\bibitem{cui2024lotka}
	S. Cui, Q. Zhao, G. Zhang, H. Jardon-Kojakhmetov, and M. Cao, ``On the analysis of a higher-order Lotka-Volterra model: An application of S-tensors and the polynomial complementarity problem,''  IEEE Transactions on Automatic Control, 2025.
	
	\bibitem{gallo1993directed}
	G. Gallo, G. Longo, S. Pallottino, and S. Nguyen, ``Directed hypergraphs and applications,'' \textit{Discrete Appl. Math.}, vol. 42, nos. 2--3, pp. 177--201, 1993.
	
	\bibitem{Ding2013}
	W. Ding, L. Qi, and Y. Wei, ``M-tensors and nonsingular m-tensors,'' \textit{Linear Algebra Appl.}, vol. 439, no. 10, pp. 3264--3278, 2013.
	
	\bibitem{Cui2024}
	S. Cui, G. Zhang, H. Jardon-Kojakhmetov, and M. Cao, ``On metzler positive systems on hypergraphs,'' arXiv preprint arXiv:2401.03652, 2024.
	
	\bibitem{wangxuezhou2019}
	X. Wang, M. Che, and Y. Wei, ``Existence and uniqueness of positive solution for $H^{+}$-tensor equations,'' \textit{Appl. Math. Lett.}, vol. 98, pp. 191--198, 2019.
	
	\bibitem{Qi2005}
	L. Qi, ``Eigenvalues of a real supersymmetric tensor,'' \textit{J. Symb. Comput.}, vol. 40, no. 6, pp. 1302--1324, 2005.
	
	
	
	
	\bibitem{Cowan68}
	J. D. Cowan, ``Statistical mechanics of nervous nets,'' in \textit{Neural Networks: Proceedings of the School on Neural Networks - June}, Berlin, Germany: Springer Berlin Heidelberg, 1968.
	
	\bibitem{Cowan70}
	J. D. Cowan, ``A statistical mechanics of nervous activity,'' in \textit{Lectures on Mathematics in the Life Sciences}, vol. 2, pp. 1--57, 1970.		
	
\end{thebibliography}
\end{document}